\documentclass[10pt, conference]{IEEEtran} 

\usepackage{listings}
\usepackage{ltl}
\usepackage{multirow}
\usepackage{tabulary}
\usepackage{xcolor, colortbl}
\usepackage{ltl}
\usepackage{tikz}
\usepackage[framemethod=tikz]{mdframed}
\usepackage{ifthen}
\usepackage{tikzscale}
\usepackage{pgfplots}
\tikzset{every picture/.style=semithick,
}

\tikzset{
    >=stealth,
    possible world/.style={circle,draw,thick,align=center},
    real world/.style={double,circle,draw,thick,align=center},
}
\usetikzlibrary{positioning,mindmap,automata,fit,backgrounds,shapes, 
arrows,shapes.misc,patterns}
\usetikzlibrary{decorations.pathmorphing}
\usetikzlibrary{decorations.markings}
\usetikzlibrary{calc}

\usepackage{ltl}
\usepackage{tabulary}
\usepackage{url}
\usepackage{todonotes}
\usepackage{url}
\usepackage{cite}
\usepackage{xspace}
\usepackage{graphicx}
\usepackage{subfigure}
\usepackage{caption}
\usepackage{color}
\usepackage{multicol}
\usepackage{hyphenat}
\usepackage{fancyvrb}
\usepackage{lipsum}
\usepackage{stmaryrd}
\usepackage{wrapfig}

\usepackage{hyperref}
\hypersetup{%
  colorlinks=true,           
  allcolors=blue!70!black,   
  pdfstartview=Fit,          
  breaklinks=true,
  pdfauthor={B. Bonakdarpour and B. Finkbeiner},
  pdftitle={Controller Synthesis for Hyperproperties}}

\usepackage[ruled,vlined,linesnumbered]{algorithm2e}

\DefineVerbatimEnvironment{code}{Verbatim}{fontsize=\small}

\begin{document}

\newcommand\modified[1]{#1}
\newcommand\bmodified[1]{#1}

\newcommand\frombaa[1]{\textcolor{red}{FROM BAA: #1}}
\newcommand\jyo[1]{\textcolor{blue}{Jyo: #1}}
\newcommand\chao[1]{\textcolor{blue}{Chao: #1}}
\newcommand\ufuk[1]{\textcolor{green}{Ufuk: #1}}
\newcommand\scott[1]{\textcolor{green}{Scott: #1}}
\newcommand\georgios[1]{\textcolor{brown}{Georgios: #1}}

\newcommand{\mypara}[1]{\vspace{0.5em} \noindent {\bf #1}.}
\newcommand{\myipara}[1]{\vspace{0.4em} \noindent {\em #1}.}

\newcommand{\lecps}{{\sc le}-{\sc cps}\xspace}
\newcommand{\lecpss}{{\sc le}-{\sc cps}\xspace}
\newcommand{\lecs}{{\sc lec}{\it s}\xspace}
\newcommand{\lec}{{\sc lec}\xspace}

\newcommand{\ignore}[1]{}

\newcommand\smin{\textcolor{red}{ins}}
\newcommand\smh{\textcolor{red}{hhs}}
\newcommand\smop{\textcolor{red}{ops}}
\newcommand\medin{\textcolor{red}{inm}}
\newcommand\medh{\textcolor{red}{hhm}}
\newcommand\medop{\textcolor{red}{opm}}
\newcommand\lgin{\textcolor{red}{inl}}
\newcommand\lgh{\textcolor{red}{hhl}}
\newcommand\lgop{\textcolor{red}{opl}}
\newcommand{\cur}{{\bf **}}
\newcommand{\current}{C}
\newcommand{\pending}{P}

\newcommand{\Phat}{\hat{P}}

\newcommand{\Paths}[2]{\mathit{Paths}^{#1}{(#2)}}
\newcommand{\fPaths}[2]{\mathit{Paths}^{#1}_{\mathit{fin}}{(#2)}}
\newcommand{\dom}{\mathit{dom}}
\newcommand{\dtmc}{\mathcal{M}}
\newcommand{\LTL}{\textsf{\small LTL}\xspace}
\newcommand{\CTL}{\textsf{\small CTL}\xspace}
\newcommand{\CTLstar}{\textsf{\small CTL$^*$}\xspace}
\newcommand{\PCTL}{\textsf{\small PCTL}\xspace}
\newcommand{\PCTLstar}{\textsf{\small PCTL$^*$}\xspace}
\newcommand{\HyperPCTL}{\textsf{\small HyperPCTL}\xspace}
\newcommand{\HyperLTL}{\textsf{\small HyperLTL}\xspace}
\newcommand{\HyperCTLstar}{\textsf{\small HyperCTL$^*$}\xspace}
\newcommand{\AFHyperLTL}{\mbox{AF-HyperLTL}\xspace}
\newcommand{\matching}{\mathcal{M}}
\newcommand{\topolgy}{\mathcal{T}}

\newcommand{\alphabet}{\mathrm{\Sigma}}
\newcommand{\states}{\mathrm{\Sigma}}
\newcommand{\statespace}{\states}
\newcommand{\Trace}{\mathsf{Traces}}
\newcommand{\trace}{t}
\newcommand{\qtrace}{\eta}
\newcommand{\sform}{\mathrm{\Phi}}
\newcommand{\pform}{\varphi}

\newcommand{\qed}{$~\blacksquare$}
\newcommand{\naturals}{\mathbb{N}_{>0}}
\newcommand{\naturalszero}{\mathbb{N}_{\geq 0}}

\newcommand{\AP}{\mathsf{AP}}

\newcommand{\Next}{\X}
\newcommand{\Finally}{\F}
\newcommand{\Globally}{\G}
\newcommand{\V}{\mathcal{V}}

\newcommand{\pr}{\mathbb{P}}
\renewcommand{\Pr}{\mathit{Pr}}

\newcommand{\emptyword}{\epsilon}

\newcommand{\init}{\mathit{init}}
\newcommand{\tpm}{\mathbf{P}}
\newcommand{\quant}{\mathbb{Q}}

\newcommand{\dbsim}{\mathit{dbSim}}
\newcommand{\res}{\mathit{res}}
\newcommand{\qout}{\mathit{qOut}}
\newcommand{\env}{\mathit{env}}
\newcommand{\fail}{\mathit{fail}}

\newcommand{\comp}[1]{\textsf{\small #1}}

\newcommand{\sigmakp}{$\mathsf{\Sigma^p_k}$\comp{-complete}\xspace}
\newcommand{\pikp}{$\mathsf{\Pi^p_k}$\comp{-complete}\xspace}

\newcommand\donotshow[1]{}

\newcommand{\ie}{i.e.\xspace}
\newcommand{\shield}{SHIELD\xspace}
\newcommand{\lestl}{\leccomp[{\sc stl}]}
\newcommand{\mulf}[1]{\multicolumn{2}{l}{#1}}

\newcommand{\z}{\cellcolor{black}}
\newcommand{\x}{\cellcolor{lightgray}}

\newcommand{\tru}{\mathsf{true}}
\newcommand{\false}{\mathsf{false}}
\newcommand{\inn}{\mathsf{in}}
\newcommand{\out}{\mathsf{out}}
\newcommand{\suffix}[2]{#1[#2,\infty]}
\newcommand{\F}{\LTLdiamond}
\newcommand{\G}{\LTLsquare}
\newcommand{\U}{\,\mathcal U\,}
\newcommand{\W}{\,\mathcal W\,}
\newcommand{\X}{\LTLcircle}

\newcommand{\States}{S}
\newcommand{\state}{s}
\newcommand{\trans}{\delta}
\newcommand{\inp}{\mathfrak{u}}
\newcommand{\uncont}{\mathfrak{u}}
\newcommand{\cont}{\mathfrak{c}}
\newcommand{\outp}{\mathcal{C}}
\newcommand{\kframe}{\mathcal{F}}
\newcommand{\krip}{\mathcal{P}}
\newcommand{\ktuple}{\langle S, s_\init, \cont, \uncont, L \rangle}
\newcommand{\ktupleprime}{\langle S', s'_\init, \cont', \uncont', L' \rangle}
\newcommand{\ktupleprimewithuc}{\langle S', s'_\init, \trans' \cup \inp, L' 
\rangle}
\newcommand{\lang}{\mathcal{L}}

\newcommand{\pos}{\mathit{pos}}
\newcommand{\negt}{\mathit{neg}}
\newcommand{\CS}[2]{\mbox{\sf \small CS[#1, #2{}]}\xspace}

\newcommand{\GMNI}{\textsf{\small GMNI}\xspace}
\newcommand{\GNI}{\textsf{\small GNI}\xspace}

\newtheorem{theorem}{Theorem}
\newtheorem{definition}{Definition}
\newtheorem{corollary}{Corollary}

\pgfdeclarelayer{background}
\pgfdeclarelayer{foreground}
\pgfsetlayers{background,main,foreground}

\tikzset{
  invisible/.style={opacity=0, text opacity=0},
  visible on/.style={alt={#1{}{invisible}}},
  alt/.code args={<#1>#2#3}{%
    \alt<#1>{\pgfkeysalso{#2}}{\pgfkeysalso{#3}} 
the path
  },
}

\newcommand\enrec[1]{%
  \tikz[baseline=(X.base)]
    \node (X) [draw, shape=circle, inner sep=0, fill=white] {\strut #1};}

\newcommand\encirclew[1]{%
  \tikz[baseline=(X.base)]
    \node (X) [draw, shape=circle, inner sep=0, fill=white] {\strut #1};}

\newcommand\encirclegr[1]{%
  \tikz[baseline=(X.base)]
    \node (X) [draw, shape=circle, inner sep=0, fill=gray] {\strut #1};}

 \newcommand\encirclep[1]{%
  \tikz[baseline=(X.base)]
    \node (X) [draw, shape=circle, inner sep=0, fill=pink] {\strut #1};}

\newcommand\encircle[1]{%
  \tikz[baseline=(X.base)]
    \node (X) [draw, shape=circle, inner sep=0, fill=green] {\strut #1};}

\newcommand\encircley[1]{%
  \tikz[baseline=(X.base)]
    \node (X) [draw, shape=circle, inner sep=0, fill=yellow] {\strut #1};}

\newcommand\encircler[1]{%
  \tikz[baseline=(X.base)]
    \node (X) [draw, shape=circle, inner sep=0, fill=red] {\strut #1};}

\tikzstyle{place}=[circle,thick,draw=blue!75,fill=blue!20,minimum size=6mm]
\tikzstyle{red place}=[place,draw=red!75,fill=red!20]
\tikzstyle{transition}=[rectangle,thick,draw=black, fill=black, 
minimum size=1mm]

\newcommand{\dec}{\mathtt{dec}}
\newcommand{\ses}{\mathtt{ses}}
\newcommand{\session}{\mathtt{session}}
\newcommand{\status}{\mathtt{status}}
\newcommand{\ntf}{\mathtt{ntf}}

\definecolor{mGreen}{rgb}{0,0.6,0}
\definecolor{mGray}{rgb}{0.5,0.5,0.5}
\definecolor{mPurple}{rgb}{0.58,0,0.82}
\definecolor{backgroundColour}{rgb}{0.95,0.95,0.92}

\lstdefinestyle{CStyle}{
    backgroundcolor=\color{backgroundColour},   
    commentstyle=\color{mGreen},
    keywordstyle=\color{magenta},
    numberstyle=\tiny\color{mGray},
    stringstyle=\color{mPurple},
    basicstyle=\footnotesize,
    breakatwhitespace=false,         
    breaklines=true,                 
    captionpos=b,                    
    keepspaces=true,                 
    numbers=left,                    
    numbersep=2pt,                  
    showspaces=false,                
    showstringspaces=false,
    showtabs=false,                  
    tabsize=2,
    language=C
}


\newcommand{\eab}[1]{{\color{red}#1}}

\renewcommand{\topfraction}{0.96}
\renewcommand{\bottomfraction}{0.95}
\renewcommand{\textfraction}{0.1}
\renewcommand{\floatpagefraction}{1}
\renewcommand{\dbltopfraction}{.97}
\renewcommand{\dblfloatpagefraction}{.99}

\title{Controller Synthesis for Hyperproperties\thanks{This research has been 
partially supported by the United States NSF SaTC Award 1813388.}}
%
%
%

\author{
\IEEEauthorblockN{Borzoo Bonakdarpour}
\IEEEauthorblockA{Department of Computer Science and Engineering\\
Michigan State University, USA\\
Email: \sf{borzoo@msu.edu}}
\and
\IEEEauthorblockN{Bernd Finkbeiner}
\IEEEauthorblockA{Reactive Systems Group\\
Saarland University, Germany\\
Email: \sf{finkbeiner@cs.uni-saarland.de}}
}

\maketitle

\begin{abstract}

We investigate the problem of {\em controller synthesis} for 
{\em hyperproperties} specified in the temporal logic HyperLTL. Hyperproperties 
are system properties that relate multiple execution traces. Hyperproperties 
can elegantly express information-flow policies like noninterference and 
observational determinism. The controller synthesis problem is to 
automatically design a controller for a {\em plant} that ensures satisfaction 
of a given specification in the presence of the environment or adversarial 
actions.
%
%
We show that the controller synthesis problem is decidable for HyperLTL 
specifications and finite-state plants. We provide a rigorous complexity 
analysis for different fragments of HyperLTL and different system types: 
tree-shaped, acyclic, and general graphs.

\end{abstract} 

\section{Introduction}
\label{sec:intro}

In {\em program synthesis}, an algorithm automatically constructs a 
program that satisfies a given high-level specification given in some formal 
logic~\cite{ec82,mw84}. The {\em controller synthesis}~\cite{rw89} problem
is a specific form of program synthesis where we ask whether the 
\emph{controllable transitions} of a given system, called the {\em plant}, can 
be selected in such a way that the resulting restricted system satisfies the 
given specification. Algorithms for this problem automate the design of a 
controller that ensures the satisfaction of the specification even in the 
presence of an {\em adversary} modeled by {\em uncontrollable transitions}. 
Synthesis guarantees correctness by construction and it enables users to 
refrain from the error-prone process of developing software and, instead, to focus on only analyzing the functional behavior of the system. Thus, program synthesis 
exhibits its power particularly in automating the generation of intricate and 
complex parts of a system.

A particular area where synthesis can play an important role is the 
construction of systems where the flow of information is critical, for example 
to ensure {\em confidentiality}. This is because even a short transient 
violation of {\em security} or {\em privacy} policies may result in leaking 
private or highly sensitive information, compromising safety, or interrupting 
vital public or social services. The hacking of the emergency system in the city 
of Dallas~\cite{dallas}, the {\em Heartbleed} error that leaked records of 4.5 
million patients~\cite{heartbleed}, the data leak at Yahoo, which resulted in 
stealing 500 million accounts~\cite{yahoo}, and the {\em Goto Fail} bug, where 
the encryption of more than 300 million devices was broken~\cite{goto}, only scratch 
the surface of high-profile examples of such security breaches. 
Synthesis with respect to security policies is of 
particular interest as it can construct protocols that are guaranteed to behave correctly in the 
presence of adversarial attacks or untrusted parties. Also, given a set 
of actions and primitives of parties that participate in a protocol, synthesis can be used to synthesize trusted third parties that mediate 
between other parties (see Section~\ref{sec:example} for a concrete example of 
an application of controller synthesis to non-repudiation protocols).

In order to express and reason about information-flow security policies, we use 
the powerful formalism of {\em hyperproperties}~\cite{cs10}. Hyperproperties 
elevate trace properties from a set of execution traces to sets of sets of 
execution traces. {\em Temporal 
logics for hyperproperties} such as HyperLTL~\cite{cfkmrs14} have been introduced to give 
clear syntax and semantics to hyperproperties. HyperLTL allows for the simultaneous quantification over the temporal behavior of \emph{multiple} 
execution traces. Atomic propositions are indexed to refer to specific traces. For example, 
\emph{noninterference}~\cite{gm82}
between a secret input $h$ and a public output $o$ can be specified in HyperLTL 
by stating that, for all pairs of traces $\pi$ and $\pi'$,  if the input is the 
same for all input variables $I$ except $h$, then the output $o$ must be the 
same at all times:
\[
\forall\pi.\forall\pi'.~ \G \big(\!\!\!\bigwedge_{i\in I\setminus \{h\}}\! i_\pi 
= i_{\pi'}\big) ~\Rightarrow~ \G\, (o_\pi = o_{\pi'})
\]
Another prominent example is \emph{generalized noninterference} 
(GNI)~\cite{m88}, which can be expressed as the following HyperLTL formula:
\[
\forall\pi.\forall\pi'.\exists\pi''.~\G\, (h_\pi = h_{\pi''}) ~\wedge~ \G\, 
(o_{\pi'} = o_{\pi''})
\]
The existential quantifier is needed to allow for nondeterminism. Generalized 
noninterference permits nondeterminism in the low-observable behavior, but 
stipulates that low-security outputs may not be altered by the injection of 
high-security inputs.

\begin{table*}[t]
\def\arraystretch{2.5}
\centering
\scalebox{1}
{\newcolumntype{K}[1]{>{\centering\arraybackslash}p{#1}}
\begin{tabular}{|K{2cm}||K{3cm}||K{4cm}||K{5cm}|}
\hline

{\bf HyperLTL fragment} & {\bf Tree} & {\bf Acyclic} & {\bf General}\\

\hline\hline
$\mbox{E}^*$ &
\multirow{2}{*}{\parbox[c]{2cm}{\centering \comp{L-complete}\\ 
{\em (Theorem~\ref{thrm:sys-tree-ea})}}} & 
{\parbox[c]{2cm}{\centering \comp{NL-complete} \\
{\em (Theorem~\ref{thrm:sys-acyc-e})}}}  &
{\parbox[c]{2cm}{\centering \comp{NL-complete}\\ {\em 
(Theorem~\ref{thm:sys-general-e}) }}} \\

\cline{1-1}\cline{3-3}\cline{4-4}

$\mbox{E}^*\mbox{A}$ &   & \parbox[c]{2cm}{\centering \comp{$\mathsf{\Sigma}_{2}^p$}\\ {\em 
(Theorem~\ref{thrm:sys-acyc-ea}) }} & \multirow{2}{*}
{\parbox[c]{3cm}{\centering \comp{PSPACE-complete}\\ {\em 
(Theorem~\ref{thrm:system-general-EAk})}}}\\
\cline{1-3}
%
%
%
$\mbox{AE}^*$ & {\parbox[c]{2cm}{\centering
\comp{P-complete} \\ {\em (Theorem~\ref{thrm:sys-tree-ae})}}}
&  
{\parbox[c]{2cm}{\centering \comp{$\mathsf{\Sigma}_{2}^p$-complete}\\ {\em 
(Theorem~\ref{thrm:system-acyc-EAk1}) }}} 
& \\
\hline
%
%
$\mbox{A}\mbox{A}^+$ & 
\multirow{4}{*}{{\parbox[c]{2cm}{\centering \comp{NP-complete} \\ {\em 
(Corollary~\ref{cor:sys-tree-aa})} }}} &
{{\parbox[c]{2cm}{\centering \comp{NP-complete} \\ {\em 
(Theorem~\ref{thrm:sys-acyc-a})} }}} &
{{\parbox[c]{2cm}{\centering \comp{NP-complete} \\ {\em 
(Theorem~\ref{thm:sys-general-a})} }}}  \\
\cline{1-1} \cline{3-4}
%
%
$(\mbox{E}^*\mbox{A}^*)^k$, $k \geq 2$ &   &
{\parbox[c]{2cm}{\centering \comp{$\mathsf{\Sigma}_{k}^p$-complete}\\ {\em 
(Theorem~\ref{thrm:system-acyc-EAk1}) }}} 
 & 
\multirow{2}{*}
{\parbox[c]{4cm}{\centering \comp{$(k{-}1)$-EXPSPACE-complete}\\ {\em 
(Theorem~\ref{thrm:system-general-EAk})}}}\\
\cline{1-1} \cline{3-3}
%
%
$(\mbox{A}^*\mbox{E}^*)^k$, $k \geq 1$ &  & 
{\parbox[c]{3cm}{\centering \comp{$\mathsf{\Sigma}_{k+1}^p$-complete}\\ {\em 
(Theorem~\ref{thrm:system-acyc-EAk1}) }}} 
 & \\
\cline{1-1} \cline{3-4}
%
%
$(\mbox{A}^*\mbox{E}^*)^*$ &  & 
{\parbox[c]{3cm}{\centering \comp{PSPACE} \\ {\em 
(Corollary~\ref{cor:sys-acyclic-hltl}) }}} & 
{\parbox[c]{2cm}{\centering \comp{NONELEMENTARY}\\ {\em 
(Corollary~\ref{cor:sys-general-hltl})}}} \\
\hline

\end{tabular}
}
\vspace{2mm}
\caption{Complexity of the HyperLTL controller synthesis problem in the size of 
the plant, where $k$ is the number of quantifier alternations in 
the formula.}
\label{tab:system}
\end{table*}

Techniques for automatically constructing systems that satisfy a given set of information-flow properties 
are still in their infancy. The general synthesis problem for HyperLTL is known to be 
undecidable as soon as the formula contains two 
universal quantifiers~\cite{DBLP:journals/acta/FinkbeinerHLST20}. To remedy this problem, {\em bounded 
synthesis}~\cite{DBLP:journals/acta/FinkbeinerHLST20,cfst19} restricts the search to implementations up 
to a given bound on the number of states.  Bounded synthesis is decidable, but still very difficult, because the transition structure  of the controller must be found. Bounded synthesis has been successfully applied to examples such as the dining cryptographers\cite{c85}, but does not seem to scale to large systems.
Another prominent approach is {\em 
program repair}~\cite{bf19}, where, for a given program that does not satisfy a 
HyperLTL property, transitions are eliminated until the program satisfies the property. Program repair has the advantage that the repair directly works on the state space of the given program. However, program repair is not reactive in the sense that it  
does not allow for modeling the actions of an adversarial environment, which is vital to dealing with information-flow security.

This limitation is addressed by \emph{controller synthesis}. Controller synthesis is based on a \emph{plant}, which describes the system behavior in terms of controllable and uncontrollable transitions.
Like program repair, controller synthesis has the advantage that the state
space is already given. Unlike program repair, controller synthesis distinguishes between controllable and uncontrollable transitions and therefore considers an adversary.
The synthesized controller guarantees that, no matter 
how the adversary behaves, the 
specification is satisfied.

In this paper, we study the controller synthesis problem of 
finite-state systems with respect to HyperLTL specifications. We provide a 
detailed analysis of the complexity of the controller synthesis problem for different 
fragments of HyperLTL and different shapes of the plants motivated by 
the following observations:

\begin{itemize}

\item The  security and privacy policies of interest vary in the quantifier structure that is needed to express the policy in HyperLTL. Typical examples are $\forall\forall$ (noninterference), $\forall\forall\exists$ (generalized noninterference), and $\forall\exists$ (noninference)~\cite{cfkmrs14}.
Data minimization, a popular 
privacy technique, is of the form $\forall\forall\exists\exists$~\cite{sssb19}. 
The non-repudiation protocol discussed in Section~\ref{sec:example} is specified with a HyperLTL formula of the form $\exists\forall\forall\forall$.

\item The protocols and systems of interest vary in the shape of their state 
graphs. We are interested in {\em general}, {\em acyclic}, and {\em 
tree-shaped} plants. The need for investigating the controller synthesis 
problem for tree-shaped and acyclic plants stems from two reasons.
First, many 
trace logs are in the form of a simple linear collection of the traces seen so 
far. Or, for space efficiency, the traces are organized by common prefixes and 
assembled into a tree-shaped graphs, or by common prefixes as well as suffixes 
assembled into an acyclic graphs. These example scenarios can be used to 
synthesize protocols, as has been done for synthesizing distributed 
algorithms~\cite{at17} that respect certain safety and liveness constraints. In 
the context of security/privacy, the example scenarios would be used to 
synthesize a complete protocol that satisfy a set of hyperproperties. The 
second reason is that, tree-shaped and acyclic graphs often occur as the 
natural representation of the state space of a protocol. For example, certain 
security protocols, such as non-repudiation, authentication, and session-based 
protocols (e.g., TLS, SSL, SIP) go through a finite sequence of \emph{phases}, 
resulting in an acyclic plant A detailed example of synthesizing a 
tree-shaped non-repudiation protocol is presented in 
Section~\ref{sec:example}.

\end{itemize}
We investigate the difficulty of the controller 
synthesis problem with respect to different combinations of the quantifier structure of the 
HyperLTL formula and the shape of the plant.

Table~\ref{tab:system} summarizes the results of this paper. The complexities are in the size 
of the plant. This \emph{system complexity} is the most relevant complexity in 
practice, because the system tends to be much larger than the specification. 
Our results show that the shape of the plant plays a crucial role in the 
complexity of the controller synthesis problem.

\begin{itemize}
\item {\bf Trees.} \ For trees, the complexity in the size of the plant does 
not go beyond \comp{NP}. The problem for the 
existential alternation-free fragment and the fragment with one quantifier 
alternation where the leading quantifier is existential is \comp{L-complete}. 
The problem for the fragment with one quantifier alternation where the leading 
quantifier is universal is \comp{P-complete}. However, the problem becomes 
\comp{NP-complete} as soon as there are two leading universal quantifiers. 

\item {\bf Acyclic graphs.} \ For acyclic plants, the complexity is 
\comp{NL-complete} for the existential fragment. Similar 
to tree-shaped graphs, the problem becomes \comp{NP-complete} as soon 
as there are two leading universal quantifiers. Furthermore, the complexity is 
in the level of the polynomial hierarchy that corresponds to the number of 
quantifier alternations.

\item {\bf General graphs.} \ For general plants, the complexity is 
\comp{NL-complete} for the existential fragment and \comp{NP-complete} for the 
universal fragment. The complexity is \comp{PSPACE-complete} for the 
fragment with one quantifier alternation and $(k-1)$-\comp{EXPSPACE-complete} 
in the number $k$ of quantifier alternations.

\end{itemize}
Surprisingly, the complexities identified in Table~\ref{tab:system} are very 
much aligned with those reported for the program repair problem~\cite{bf19}. 
The main exceptions are the complexity for the universal 
alternation-free fragment in tree-shaped and acyclic graphs that are 
\comp{NP-complete} in the case of controller synthesis and are 
\comp{L-complete} and \comp{NL-complete}, respectively, in 
case of the repair problem. This fragment is of particular interest, as it 
hosts many of the important information-flow security policies.   

We believe that the results of this paper provide the fundamental understanding 
of the controller synthesis problem for secure information flow and pave the way for 
further research on developing efficient and scalable techniques.

\paragraph*{Organization} The remainder of this paper is organized as follows.
In Section~\ref{sec:prelim}, we review HyperLTL. We present a detailed 
motivating example in Section~\ref{sec:example}. The formal statement of the 
controller synthesis problem is in Section~\ref{sec:problem}.
Section~\ref{sec:tree} presents our results on the complexity of 
controller synthesis for HyperLTL in the size of tree-shaped plants. 
Sections~\ref{sec:acyclic} and~\ref{sec:general} present the results on the 
complexity of synthesis in acyclic and general graphs, respectively. We discuss 
related work in Section~\ref{sec:related} and conclude with a discussion of 
future work in Section~\ref{sec:conclusion}.

\section{Preliminaries}
\label{sec:prelim}


\subsection{Plants}
\label{subsec:krip}

Let $\AP$ be a finite set of {\em atomic propositions} and $\alphabet = 
2^{\AP}$ be the {\em alphabet}. A {\em letter} is an element of $\alphabet$. A 
\emph{trace} $t\in \Sigma^\omega$ over alphabet $\alphabet$ is an infinite 
sequence of letters: $t = t(0)t(1)t(2) \ldots$


\begin{definition}
\label{def:kripke}
A {\em plant} is a tuple $\krip = \ktuple$, where 

\begin{itemize}
 \item $\States$ is a finite set of {\em states};
 
 \item $\state_{\init} \in \States$ is the {\em initial state};
 
 \item $\cont,\uncont \subseteq \States \times \States$ are respectively sets 
of  of {\em controllable} and {\em uncontrollable transitions}, where $\cont 
\cap \uncont = \{\}$, and  

 \item $L: S \rightarrow \statespace$ is a {\em labeling function} on the 
states of $\krip$.
\end{itemize}
We require that for each $\state \in \States$, there exists $\state' \in 
\States$, such that $(\state, \state') \in \cont \cup \uncont$.

\end{definition}

Figure~\ref{fig:kripke} shows an example plant, where transition 
$(s_\init, s_1)$ is an uncontrollable transition while the rest are 
controllable, and $L(s_{\init})= \{a\}, L(s_3)=\{b\}$, etc. The 
\emph{size} of the plant is the number of 
its states. The directed graph $\kframe = \langle \States, \cont \cup \uncont 
\rangle$ is 
called the {\em frame} of the plant $\krip$. A {\em loop} in 
$\kframe$ is a finite sequence $\state_0\state_1\cdots \state_n$, such that 
$(\state_i, \state_{i+1}) \in \cont \cup \uncont$, for all $0 \leq i < n$, and 
$(\state_n, \state_0) \in \cont \cup \uncont$. We call a frame {\em acyclic}, 
if the 
only loops are self-loops on terminal states, i.e., on states that 
have no other outgoing transition.

\begin{figure}[t]
\centering
\scalebox{0.8}{
\begin{tikzpicture}[-,>=stealth',shorten >=.5pt,auto,node distance=2cm, 
semithick, initial text={}]

\node[initial, state] [text width=1em, text centered, minimum 
  height=2.5em](0) {\hspace*{-1.25mm}$\{a\}$};

\node [above left = 0.005 cm and 0.1 cm of 0](label){$s_{\init}$};

\node[state, above right=of 0][text width=1em, text centered, minimum 
  height=2.5em] (1) {\hspace*{-1.25mm}$\{a\}$};

\node [above left = 0.005 cm and 0.1 cm of 1](label){$s_{1}$};

\node[state, right=of 1][text width=1em, text centered, minimum 
height=2.5em] (2) {\hspace*{-1.25mm}$\{b\}$};

\node [above left = 0.005 cm and 0.1 cm of 2](label){$s_{2}$};

\node[state, right=of 0][text width=1em, text centered, minimum 
height=2.5em] (3) {\hspace*{-1.25mm}$\{b\}$};

\node [above left = 0.005 cm and 0.1 cm of 3](label){$s_{3}$};

\draw[->, dashed]   
   (0) edge node (01 label) {} (1);

\draw[->]   
  (0) edge node (03 label) {} (3)
  (1) edge node (12 label) {} (3)
  (1) edge node (12 label) {} (2)
  (2) edge [loop right] node (22 label) {} (2)
  (3) edge [loop right] node (33 label) {} (3);
    
\end{tikzpicture}
}
\caption {An acyclic plant.}
\label{fig:kripke}
\end{figure}
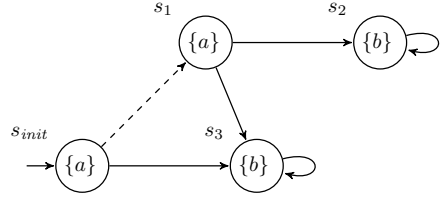

We call a frame \emph{tree-shaped}, or, in short, a \emph{tree}, if 
every state $\state$ has a unique state $\state'$ with $(\state', \state) \in 
\cont \cup \uncont$, except for the root node, which has no predecessor, and 
the leaf nodes, which, again because of Definition~\ref{def:kripke}, 
additionally have a self-loop but no other outgoing transitions. In some cases, 
the system at hand is given as a tree-shaped or acyclic plant. 
Examples include session-based security protocols and space-efficient 
execution logs, because trees allow us to organize the traces according to 
common prefixes and acyclic graphs according to both common prefixes and common 
suffixes.

A \emph{path} of a plant is an infinite sequence of states
$\state(0)\state(1)\cdots \in \States^\omega$, such that:

\begin{itemize}
 \item $\state(0) = \state_\init$, and
\item $(\state(i), \state({i+1})) \in \cont \cup \uncont$, for all $i \geq 0$. 
\end{itemize}
A trace of a plant is a trace $t(0)t(1)t(2) \cdots \in 
\alphabet^\omega$, such that there exists a path $\state(0)\state(1)\cdots \in 
\States^\omega$ with $t(i) = L(\state(i))$ for all $i\geq 0$. We denote by 
$\Trace(\krip)$ the set of all traces of $\krip$ with paths that start 
in state $\state_\init$.

\subsection{The Temporal Logic HyperLTL}
\label{subsec:hltl}

HyperLTL~\cite{cfkmrs14} is an extension of linear-time temporal logic (LTL) for hyperproperties.
The syntax of HyperLTL formulas is defined inductively by the following grammar:
\begin{equation*}
\begin{aligned}
& \varphi ::= \exists \pi . \varphi \mid \forall \pi. \varphi \mid \phi \\
& \phi ::= \tru \mid a_\pi \mid \lnot \phi \mid \phi \vee \phi \mid \phi \
\U \, \phi \mid \X \phi
    \end{aligned}
\end{equation*}
where $a \in \AP$ is an atomic proposition and $\pi$ is a trace variable from 
an infinite supply of variables $\V$. The Boolean connectives $\neg$ and 
$\vee$ have the usual meaning, $\U$ is the temporal \emph{until} operator and 
$\X$ is the temporal \emph{next} operator. We also consider the usual derived 
Boolean connectives, such as $\wedge$, $\Rightarrow$, and $\Leftrightarrow$, 
and the derived temporal operators \emph{eventually} $\F\varphi\equiv 
\tru\,\U\varphi$ and \emph{globally} $\G\varphi\equiv\neg\F\neg\varphi$.
The quantified formulas $\exists \pi$ and $\forall \pi$ are read as `along some 
trace $\pi$' and `along all traces $\pi$', respectively. For example, the 
following formula:
$$
\forall \pi.\forall \pi'. \G(a_\pi \Leftrightarrow a_{\pi'})
$$
intends to express that every pair of traces should always agree on the 
position of proposition $a$.

The semantics of HyperLTL is 
defined with respect to a trace assignment, a partial mapping~$\Pi \colon \V 
\rightarrow \alphabet^\omega$. The assignment with empty domain is denoted by 
$\Pi_\emptyset$. Given a trace assignment~$\Pi$, a trace variable~$\pi$, and 
a 
concrete trace~$t \in \alphabet^\omega$, we denote by $\Pi[\pi \rightarrow 
t]$ 
the assignment that coincides with $\Pi$ everywhere but at $\pi$, which is 
mapped to trace $t$. Furthermore, $\suffix{\Pi}{j}$ denotes the assignment 
mapping each trace~$\pi$ in $\Pi$'s domain to 
$\Pi(\pi)(j)\Pi(\pi)(j+1)\Pi(\pi)(j+2) \cdots $.
The satisfaction of a HyperLTL formula $\varphi$ over a trace assignment 
$\Pi$ and a set of traces $T \subseteq \alphabet^\omega$, denoted by $T,\Pi 
\models \varphi$, is defined as follows:
\[
\begin{array}{l@{\hspace{2em}}c@{\hspace{2em}}l}
  T, \Pi \models \tru\\
  T, \Pi \models a_\pi & \text{iff} & a \in \Pi(\pi)(0),\\
  T, \Pi \models \neg \psi & \text{iff} & T, \Pi \not\models \psi,\\
  T, \Pi \models \psi_1 \vee \psi_2 & \text{iff} & T, \Pi \models \psi_1 
\text{ 
or } T, \Pi \models \psi_2,\\
  T, \Pi \models \X \psi & \mbox{iff} & T,\suffix{\Pi}{1} \models \psi,\\
  T, \Pi \models \psi_1 \U \psi_2 & \text{iff} &  \exists i \ge 0: 
T,\suffix{\Pi}{\modified{i}} \models \psi_2 \ \wedge\ \\
& & \forall j \in [0, i): T,\suffix{\Pi}{j} \models \psi_1,\\
  T, \Pi \models \exists \pi.\ \modified{\psi} & \text{iff} & \exists t \in T: T,\Pi[\pi 
\rightarrow t] \models \psi,\\
  T, \Pi \models \forall \pi.\ \modified{\psi} & \text{iff} & \forall t \in T: T,\Pi[\pi 
\rightarrow t] \models \psi.
  \end{array}
\]
We say that a set $T$ of traces satisfies a sentence~$\varphi$, denoted by $T 
\models \phi$, if $T, \Pi_\emptyset \models \varphi$. If the set $T = 
\Trace(\krip)$ is generated by a plant $\krip$, we write $\krip 
\models \varphi$.


\section{Motivating Example}
\label{sec:example}

\newcommand\Act{\mathit{Act}}
\newcommand\Obs{\mathit{Obs}}
\newcommand\NRR{\mathit{NRR}}
\newcommand\NRO{\mathit{NRO}}

In order to motivate the decision problem described in 
Section~\ref{sec:problem}, we consider the application of controller synthesis 
to construct a \emph{trusted third party} for a \emph{fair non-repudiation 
protocol}. The purpose of a non-repudiation protocol is to allow two parties to 
exchange a message without any party being able to deny having participated in 
the exchange. For this purpose, the recipient of the message obtains a 
\emph{non-repudiation of origin} (NRO) evidence and the sender of the message 
obtains a \emph{non-repudiation of receipt} (NRR) evidence. The protocol is 
\emph{effective} if it is possible to successfully transmit the message to the 
recipient and the evidence to both parties. The protocol is \emph{fair} if it is 
furthermore \emph{impossible} for one party to obtain the evidence without the 
other party \emph{also} receiving the evidence.  
Fairness in an effective protocol cannot be obtained without an external trusted 
agent, called the \emph{trusted third party}, which mediates the message 
exchange and ensures the delivery of the evidence. We now consider the problem 
of synthesizing such a trusted third party from the specification of fairness 
and effectiveness.

Let $A$ be the sender of the message, $B$ be the receiver, and $T$ the trusted 
third party. $A$ has 5 possible actions, corresponding to sending the message 
$m$ or the NRO to either $B$ or to $T$, or, alternatively, doing nothing at all: 
\begin{align*}
\Act_A=\{ & A{\rightarrow} B{:}m, A{\rightarrow} T{:}m,\\
& A{\rightarrow} B{:}\NRO, A {\rightarrow} T{:}\NRO, A{:}\mathit{skip}\}.
\end{align*}
Likewise, $B$ can send the NRR to $A$ or $T$, or do nothing at all: 
$$\Act_B=\{B{\rightarrow} A{:} \NRR, B {\rightarrow} T{:} \NRR, 
B{:}\mathit{skip}\}.
$$
$T$ can send the NRR to $A$, the NRO or $m$ to $B$, or 
do nothing at all:
$$
\Act_T=\{T{\rightarrow} A{:}\NRR, T{\rightarrow} B{:}\NRO, 
T{\rightarrow} B{:}m, T{:}\mathit{skip}\}.
$$
We'll assume that the three parties take turns and interact in a fixed number of 
rounds. The plant, thus, is {\em tree-shaped}, where along each branch the 
states belong first to $A$, then to $T$, then to $B$, and then again $A$, then 
$T$, then $B$, etc. States that belong to $A$ branch according to the actions 
in $\Act_A$, likewise states that belong to $T$ branch according to $\Act_T$, 
and states that belong to $B$ according to  $\Act_B$. We label the states with 
the atomic propositions
$$
\AP=\{m, \NRR, \NRO\} \, \cup \, \Act_A \cup \Act_T \cup 
\Act_B,
$$
where $m$ indicates that $B$ has received the message, $\NRO$ that 
$B$ has received the NRO, $\NRR$ that $A$ has received the NRR, and one of the 
actions in $\Act_A \cup \Act_T \cup \Act_B$ that the respective action has just 
occurred. Since we are interested in synthesizing the trusted third party, the 
outgoing transitions of states belonging to $T$ are controllable, all other 
transitions are uncontrollable. That is, we are only allowed to manipulate the 
behavior of the trusted third party and not the original participants in the 
protocol.

We specify \emph{effectiveness} by requiring that there is a sequence of actions 
$\pi$ such that the message, the NRR, and the NRO get received. For 
\emph{fairness}, we additionally require that if either $A$ executes the actions 
according to $\pi$ (and $B$ behaves arbitrarily) or $B$ executes $\pi$ (and $A$ 
behaves arbitrarily), then it must still hold that the NRR gets received if and 
only if the NRO gets received.
{\small
\[
\begin{array}{ll} \varphi = \exists \pi. \forall \pi'.\ (\LTLdiamond m_\pi) 
\wedge (\LTLdiamond \NRR_\pi) \wedge (\LTLdiamond \NRO_{\pi}) &\\
\multicolumn{2}{r}{\mbox{\normalsize (effectiveness)}}\\
\wedge\, \Big((\LTLsquare \bigwedge_{a \in \Act_A} a_\pi \Leftrightarrow 
a_{\pi'}) \Rightarrow \big((\LTLdiamond \NRR_{\pi'}) \Leftrightarrow 
(\LTLdiamond \NRO_{\pi'})\big)\Big)\\
\multicolumn{2}{r}{\mbox{\normalsize (fairness for $A$)}} \\
\wedge\, \Big((\LTLsquare \bigwedge_{a \in \Act_B} a_\pi \Leftrightarrow 
a_{\pi'}) \Rightarrow \big((\LTLdiamond \NRR_{\pi'}) \Leftrightarrow 
(\LTLdiamond \NRO_{\pi'})\big)\Big)\\
\multicolumn{2}{r}{\mbox{\normalsize (fairness for $B$)}} \\
\end{array}
\]
}

The following trusted third party is a correct solution to the controller 
synthesis problem:

\noindent
$T_{\mbox{correct}}:$\smallskip\\
\mbox{\qquad} (1) skip until $A{:}m {\rightarrow} T$;\\
\mbox{\qquad} (2) skip until $A{:}\NRO {\rightarrow} T$;\\
\mbox{\qquad} (3) $T{\rightarrow} B{:}m$;\\
\mbox{\qquad} (4) skip until $B {\rightarrow} T{:}\NRR$;\\
\mbox{\qquad} (5) $T {\rightarrow} B{:}\NRO$;\\
\mbox{\qquad} (6) $T {\rightarrow} A{:}\NRR$.

An \emph{incorrect} solution to the controller synthesis problem would, for 
example, be a trusted third party that does not wait for the NRR from $B$ before 
forwarding the NRO from $A$ to $B$: now $B$ could quit the protocol without ever 
providing the NRR.

\noindent
$T_{\mbox{incorrect}}:$\smallskip\\
\mbox{\qquad} (1) skip until $A{:}m {\rightarrow} T$;\\
\mbox{\qquad} (2) skip until $A{:}\NRO {\rightarrow} T$;\\
\mbox{\qquad} (3) $T{\rightarrow} B{:}m$;\\
\mbox{\qquad} (4) $T {\rightarrow} B{:}\NRO$;\\
\mbox{\qquad} (5) skip until $B {\rightarrow} T{:}\NRR$;\\
\mbox{\qquad} (6) $T {\rightarrow} A{:}\NRR$.
  
$T_{\mbox{incorrect}}$ violates the fairness requirement for $A$.
      Note, however, that $\varphi$ also admits the following, somewhat 
counterintuitive, solution $T_{\mbox{strange}}$:

\noindent
$T_{\mbox{strange}}:$\smallskip\\
\mbox{\qquad} (1) skip until $A{:}m {\rightarrow} B$;\\
\mbox{\qquad} (2) $T {\rightarrow} B{:}\NRO$;\\
\mbox{\qquad} (3) $T {\rightarrow} A{:}\NRR$.    

In this solution, $T$ transmits the NRO and NRR even though the message $m$ was 
sent directly from $A$ to $B$ \emph{without} ever going through $T$. The reason, 
why $T_{\mbox{strange}}$ can do this, is that it can choose its actions based on 
complete information, i.e., based on the position in the tree. If we wish to 
restrict the possible solutions for $T$ to only those that are based only on the 
messages actually received by $T$, we need to add a consistency condition for 
incomplete information:
\[
\forall \pi. \forall \pi'.\ \Big(\LTLsquare \bigwedge_{o \in \Obs_T} o_\pi 
\leftrightarrow o_{\pi'}\Big) \, \Rightarrow \, \Big(\LTLsquare  \bigwedge_{a 
\in \Act_T} a_\pi \Leftrightarrow a_{\pi'}\Big)
\]
where
$$
\Obs_T = \{ A{\rightarrow} T{:}m, A{\rightarrow} T{:}\NRO, B{\rightarrow} 
T{:}\NRR \}
$$
consists of the actions that send a message to $T$.
  
Constructing a trusted third party for fair non-repudiation thus requires 
solving a controller synthesis problem for a tree-shaped plant. For 
effectiveness and fairness, a HyperLTL formula with quantifier prefix $\exists 
\pi. \forall \pi'$  is required, for consistency under incomplete information, 
additionally a HyperLTL formula with quantifier prefix $\forall \pi. \forall 
\pi'$.

\section{Problem Statement}
\label{sec:problem}

The {\em controller synthesis problem} is the following decision problem. Let 
$\krip = \ktuple$ be a plant and $\varphi$ be a closed HyperLTL 
formula, where $\krip$ may or may not satisfy $\varphi$. Does there exist a 
plant $\krip' = \ktupleprime$ such that:


\begin{itemize}
 \item $\States' = \States$,
 \item $s'_\init = s_\init$,
 \item $\cont' \subseteq \cont$,
 \item $\uncont' = \uncont$,
 \item $L' = L$, and
 \item $\krip' \models \varphi$?
\end{itemize}
In other words, the goal of the controller synthesis problem is to identify a 
plant $\krip'$, whose set of traces is a subset of the traces of 
$\krip$ that satisfies $\varphi$, by only restricting the controllable 
transitions and without removing any uncontrollable transitions. Note that 
since the witness to the decision problem is a plant, following 
Definition~\ref{def:kripke}, it is 
implicitly implied that in $\krip'$, for every state $s \in \States'$, there 
exists a state $s'$ such that $(s, s') \in \cont' \cup \uncont'$. I.e., the synthesis does not create a {\em deadlock} state. 

We use the following notation to distinguish the different 
variations of the problem:
\begin{center}
 \CS{\sf Fragment}{\mbox{\sf Frame Type}},
\end{center}
where

\begin{itemize}

 \item \comp{CS} is the {\em controller synthesis} decision problem as 
described above;
 
\item \comp{Fragment} is one of the following for $\varphi$:

\begin{itemize}

 \item We use regular expressions to denote the order and pattern of repetition 
of quantifiers. For example, \comp{E$^*$A$^*$-HyperLTL} denotes the fragment, 
where an arbitrary (possibly zero) number of existential quantifiers is 
followed by an arbitrary (possibly zero) number of universal quantifiers. Also, 
\comp{$\mbox{A}\mbox{E}^+$-HyperLTL} means a lead universal quantifier followed 
by one or more existential quantifiers.
%
\comp{E$^{\leq 1}$A$^*$-HyperLTL} denotes the fragment, where zero or one 
existential quantifier is followed by an arbitrary number of universal 
quantifiers.

\item \comp{(EA)$k$-HyperLTL}, for $k\geq 0$, denotes the fragment with $k$ 
alternations and a lead existential quantifier, where $k=0$ means an 
alternation-free formula with only existential quantifiers;

\item \comp{(AE)$k$-HyperLTL}, for $k\geq 0$, denotes the fragment with $k$ 
alternations and a lead universal quantifier, where $k=0$ means an 
alternation-free formula with only universal quantifiers,

\item \comp{HyperLTL} is the full logic HyperLTL, and

\end{itemize} 

\item \comp{Frame Type} is either \comp{tree}, \comp{acyclic}, or 
\comp{general}.

\end{itemize}

\section{Complexity of Controller Synthesis for Tree-shaped Graphs}
\label{sec:tree}

We begin by analyzing the complexity of the controller synthesis problem 
for tree-shaped plants.

\subsection{The \comp{E$^*$A} Fragment}

Our first result is that the controller synthesis problem for tree-shaped
plants can be solved in logarithmic time in the size of the plant 
for the fragment with only one quantifier alternation, where the 
leading quantifier is existential and there is only one universal 
quantifier. This fragment is the least expensive to deal with in tree-shaped 
plants and, interestingly, the complexity is the same as for the model 
checking~\cite{bf18} and model repair problems~\cite{bf19}.

\begin{theorem}
  \label{thrm:sys-tree-ea}
\CS{E$^*$A-HyperLTL}{\mbox{tree}} is \comp{L-complete} in the size of the 
plant.
\end{theorem}

\begin{proof}
We first show membership to \comp{L}. We note that the number of 
traces in a tree is bounded by the number of states, i.e., the size of the 
plant. The synthesis algorithm enumerates (using the complete plant) all 
possible assignments for the existential trace quantifiers. For each 
existential trace variable, we need a counter up to the number of traces, which 
requires only a logarithmic number of bits in the size of the plant.

For the universal quantifier, the algorithm first checks if assigning one of the 
traces that already have been assigned to the existential quantifiers now to the 
universal quantifier will violate the formula. If so, the existential assignment 
is disregarded. Otherwise, the algorithm proceeds to check if the traces that 
must implicitly be present (because of uncontrollable transitions or because of 
potential deadlocks) satisfy the formula. For this purpose, the algorithm 
evaluates the nodes of the tree bottom-up, such that a controller exits iff the 
root node evaluates positively.

Let $n$ be a node in the tree whose children are all leaves. If all outgoing transitions are controllable, then $n$ evaluates positively iff the assigning the trace of one the children to the universally quantified variable satisfies the formula. If at least one of the transitions is uncontrollable, then $n$ evaluates positively iff all such assignments satisfy the formula.

Now, let $m$ be a node further upward in the tree. If all outgoing transitions 
are controllable, then the node evaluates to true iff some child evaluates to 
true. We evaluate bottom-up the first child by moving to the first leaf node 
reachable from the first child. If the evaluation is negative, we move to the 
first leaf node reachable from the second child etc. If the evaluation of one of 
the children is positive we proceed to $m$'s parent with a positive evaluation. 
If we have reached the last child with a negative evaluation, we proceed to the 
parent with a negative evaluation.

If there is at least one uncontrollable transition, we disregard the controllable transitions and only step through the uncontrollable transitions, initially by moving to the first leaf reachable through the first uncontrollable transition.
If the evaluation is positive, we move to the first leave node reachable from the second child etc. If the evaluation of one of the children is negative we proceed to the parent with a negative evaluation. If we have reached the last child with a positive evaluation, we proceed to $m$'s parent with a positive evaluation.

The algorithm terminates when the root node has been evaluated. 
During the entire bottom-up traversal, we only need to store a pointer to a 
single node of the tree in memory, which can be done with logarithmically many 
bits.

In order to show completeness, we prove that the controller synthesis problem 
for the existential fragment is \comp{L-hard}. The \comp{L}-hardness for 
\CS{E$^*$-HyperLTL}{\mbox{tree}} follows from the \comp{L-hardness} of 
ORD~\cite{et97}. ORD is the graph-reachability problem for directed 
line graphs. Graph reachability \linebreak from $s$ to $t$ can be checked
with the synthesis problems for $\exists \pi .\ \F (s_\pi \wedge \F t_\pi)$ or
$\forall \pi .\ \F (s_\pi \wedge \F t_\pi)$.
\end{proof}

\subsection{The \comp{AE}$^*$ Fragment}


We now study the complexity of the controller synthesis problem for the 
fragment with only one quantifier alternation, where the leading quantifier is 
universal.

\begin{theorem}
  \label{thrm:sys-tree-ae}
\CS{AE$^*$-HyperLTL}{\mbox{tree}} is \comp{P-complete} in the size of the 
plant.
\end{theorem}

\begin{proof}
We show membership to \comp{P} with a simple marking algorithm. Let $\varphi = 
\forall \pi_1. \exists \pi_2 .\ \psi$. We begin by marking all leaves. We then 
proceed in several rounds, such that in each round, at least one mark is 
removed. We, hence, terminate within linearly many rounds in the size of the 
tree-shaped plant.
  
In each round, we go through all marked leaves $v_1$ and instantiate $\pi_1$ 
with the trace leading to $v_1$. We then again go through all marked leaves 
$v_2$ and instantiate $\pi_2$ with the trace leading to $v_2$, and check $\psi$ 
on the pair of traces, which can be done in linear time~\cite{bf18}. If the 
check is successful for some instantiation of $\pi_2$, we leave $v_1$ marked, 
otherwise we remove the mark. If no mark was removed by the end of the round, 
we terminate. Each round of the marking algorithm takes linear time in the size 
of the tree, the complete algorithm thus takes quadratic time. Once the marking 
algorithm has terminated, we remove all branches of the tree that are not 
marked. As established by the final round of the marking algorithm, the 
remaining tree satisfies $\varphi$. For additional existential quantifiers, we 
go in each round through the possible instantiations of all the existential 
quantifiers, which can be done in polynomial time. In case we reach a 
situation, where the only way to satisfy the formula is to remove an 
uncontrollable transition, then the answer to the synthesis problem is negative.

For the lower bound, we reduce HORN-SAT, which is \comp{P-complete}, to the 
synthesis problem for \comp{AE$^*$} formulas. HORN-SAT is the following problem:

\begin{quote}

Let $X=\{x_1,x_2, \ldots, x_n\}$ be a set of propositional variables. A 
{\em Horn clause} is a clause over $X$ with at most one positive literal.
Is $y = y_1 \wedge y_2 \wedge \cdots \wedge y_m$, where each $y_j$ is a Horn 
clause for all $j \in [1,m]$, satisfiable? That is, does there exist an 
assignment of truth values to the variables in $X$, such that all clauses of 
$y$ evaluate to true?

\end{quote}
In the following, we work with a modified version of HORN-SAT, where every 
clause consists of two negative and one positive literals. In order to 
transform any arbitrary Horn formula $y$ to another one $y'$ that consists of 
two negative and one positive literals, we apply the following:

\begin{itemize}
\item  To ensure that every clause contains a positive literal, we introduce a 
fresh variable $\bot$ with the intended meaning ``false''. We add $\bot$ as a 
positive literal to all clauses that have no positive literal.

\item To ensure that every clause contains at least two negative literals, we 
introduce a fresh variable $\top$ with the intended meaning ``true''. We add 
$\top$ as a negative literal to all clauses that have no negative literals. 
(Clauses with only one negative literal count as clauses with two negative 
literals with two identical negative literals.)

\item To ensure that no clause contains more than two negative literals, we 
reduce the number of negative literals as follows: Let $l_1$ and $l_2$ be two 
negative literals in a clause with more than two negative literals. We 
introduce a fresh variable $f$ and replace $l_1$ and $l_2$ with $\neg f$; we 
furthermore add $\{l_1, l_2, f\}$ as a new clause.
  
\item In order to account for the intended meaning of $\top$ and $\bot$, we 
modify the HORN-SAT problem to check if there exists a truth assignment to the 
variables in $X \cup \{\top,\bot\}$ (union fresh variables $f$ to break 
clauses as described above), such that all clauses in $y$ evaluate to 
true \emph{and} $\bot$ evaluates to false and $\top$ evaluates to true.

\end{itemize}
For example, we transform Horn formula:
\begin{align*}
y = (\neg x_1 \vee \neg x_2 \vee \neg x_3 \vee x_4) \, \wedge \, 
(\neg x_2 \vee x_4) \, \wedge \, (\neg x_1)
\end{align*}
to the following:
\begin{align*}
y' = & (\neg x_1 \vee \neg x_2 \vee f) \, \wedge \\
& (\neg x_3 \vee \neg f \vee x_4) \, \wedge \, (\neg x_2 \vee \neg 
x_2 \vee x_4) \, \wedge \, (\neg x_1 \vee \neg x_1   \vee \bot )
\end{align*}
Hence, the set of propositional variables for the transformed formula $y'$ is 
updated to $X = \{\bot, x_1, x_2, x_3, x_4, f, \top\}$. Since the 
modified 
problem and the original problem are obviously equivalent, it follows that the 
modified problem is \comp{P-complete} as well. We now describe our mapping (see 
Fig.~\ref{fig:system-tree-ae-app} for an example).

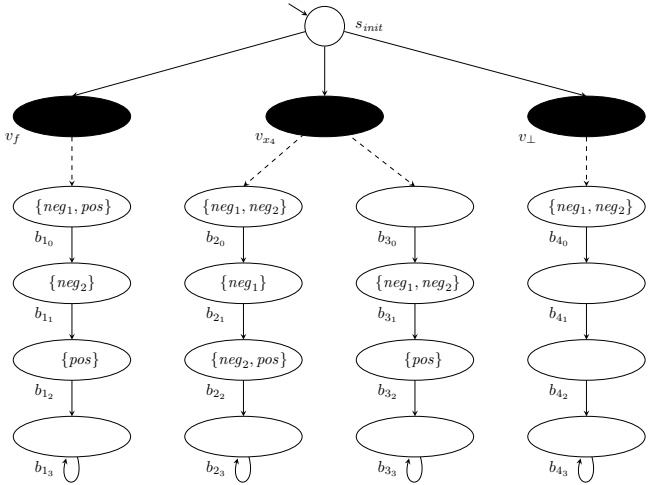
\begin{figure}[t]
\centering
\scalebox{.6}{
 \begin{tikzpicture}
 \tikzset{
    >=stealth,
    auto,
    possible world/.style={circle,draw,thick,align=center},
    real world/.style={double,circle,draw,thick,align=center},
    minimum size=25pt
}

\coordinate (init) at (0, 0);

\node[draw,circle,text width=0.5cm,fill=white] (initstate) at ($ (init) + (3, 
1) $) {};

\draw[->] ($ (initstate) + (-.8,.5) $ ) -- (initstate);


\draw [fill=black, align=center] ($ (init) + (3, -1)$) ellipse (1.3cm and 
.45cm) node [color=white](x4) {};

\draw [fill=black, align=center] ($ (init) + (8.8, -1)$) ellipse (1.3cm and 
.45cm) node [color=white](bot) {};

\draw [fill=black, align=center] ($ (init) + (-2.6, -1)$) ellipse (1.3cm and 
.45cm) node [color=white](f) {};

\node (i) at (initstate)[xshift=1cm] {$s_\init$};

\node (b) at (f) [xshift=-1.3cm,yshift=-5mm] {$v_{f}$};
\node (b) at (x4)  [xshift=-1.3cm,yshift=-5mm] {$v_{x_4}$};
\node (b) at (bot)[xshift=-1.3cm,yshift=-5mm] {$v_{\bot}$};

\foreach \i in {0,1,2,3}{
\foreach \j in {0,1,2,3}{
\draw [fill=white, align=center] ($ (init) + (-2.6, -2) + ( \i*3.8,{(\j*-1.7) 
-1} ) $) ellipse (1.3cm and .45cm) node (v\i\j)
{\ifthenelse{\i=1 \AND \j=1}{$\{\negt_1\}$}{
\ifthenelse{\i=0 \AND \j=1}{$\{\negt_2\}$}{
\ifthenelse{\i=3 \AND \j=2}{}{
\ifthenelse{\i=1 \AND \j=2}{$\{\negt_2,\pos\}$}{
\ifthenelse{\(\i=0 \AND \j=2\) \OR \(\i=2 \AND \j=2\)}{$\{\pos\}$}{
\ifthenelse{\i=3 \AND \j=0}{\hspace{-3mm}$\{\negt_1,\negt_2\}$}{
\ifthenelse{\i=3 \AND \j=3}{}{
\ifthenelse{\i=2 \AND \j=1}{$\{\negt_1, \negt_2\}$}{}
\ifthenelse{\i=1 \AND \j=0}{\hspace{-8mm}$\{\negt_1, \negt_2\}$}{
\ifthenelse{\i=3 \AND \j=1}{}{
\ifthenelse{\i=0 \AND \j=0}{\hspace{-8mm}$\{\negt_1,\pos\}$}{}
}
}}}}}}}
}};

}
}

\draw [->] (initstate) -- (f.north);
\draw [->] (initstate) -- (bot.north);
\draw [->] (initstate) -- (x4.north);

\draw [->, dashed] (x4) -- (v10.north);
\draw [->, dashed] (x4) -- (v20.north);
\draw [->, dashed] (f) -- (v00.north);
\draw [->, dashed] (bot) -- (v30.north);


\foreach \i in {0,1,2,3}{
\foreach \j/\n in {0/1,1/2,2/3}{
\draw [->] (v\i\j) -- (v\i\n);

}}
 

\foreach \i in {0,1,2,3}{
\foreach \j/\n in {0,1,2,3}{

\pgfmathsetmacro\ii{int(\i+1)};
\node (b) at (v\i\j)[xshift=-6mm,yshift=-7mm] 
{$b_{\ii_\j}$};

}}
 

\foreach \i in {0,1,2,3} {
\path (v\i3) edge [loop below] (v\i3);
}

\end{tikzpicture}
 }
 \caption{The plant for Horn formula $y = (\neg x_1 \vee \neg 
x_2 \vee \neg x_3 \vee x_4) \, \wedge \, (\neg x_2 \vee x_4) \, \wedge \, (\neg 
x_1)$.
}
 \label{fig:system-tree-ae-app}
\end{figure}

\noindent \textbf{Plant.} \ We translate the (modified) HORN-SAT 
problem to a tree-shaped plant $\krip = \ktuple$ as follows. Our general idea 
is to create a tree, where each branch represents a clause in the input
HORN-SAT problem and states of each branch are labeled according to the 
bitstring encoding of the literals in that clause. Thus, the length of each 
branch needs to be $\log(|X|)$. We now present the details:

\begin{itemize}

\item {\em (Atomic propositions $\AP$)} We include atomic propositions 
$\negt_1$ and $\negt_2$ to indicate negative literals and $\pos$ for the 
positive literals in a Horn clause. Thus,
$$\AP = \{\negt_1, \negt_2, \pos\}.$$

\item {\em (Set of states $\States$)} We now identify the members of $\States$:

\begin{itemize}
\item First, we include an initial state $\state_\init$, which is labeled with 
the empty set of atomic propositions.

\item For each propositional variable $x \in X$ that appears as a positive 
literal in a clause, we add a state $v_x$. The idea here is to 
ensure that if a positive literal appears on some clause in the synthesized 
plant, then all clauses with the same positive literal must be preserved during 
synthesis. These state are not labeled by any atomic propositions in $\AP$. 

\item For each clause $y_j$, where $j \in [1, m]$, we include a bitstring that 
represents which literals participate in $y_j$. That is, we include the 
following in $S$:
$$\Big\{b_{j_i} \mid j \in [1,m] \wedge i \in [0, \log(|X|))\Big\}.$$

We represent literals in each clause by labeling the states of the clause 
according to the appearance of the propositional variables 
$x_i \in X$ (i.e., the updated set including $\top$, $\bot$, and fresh $f$ 
variables), where $i \in [0, \log(|X|))$, as a negative or positive literal. 
More specifically, let $y_j = \{\neg x_{n_1} \vee \neg x_{n_2} \vee 
x_{p}\}$ be a Horn clause. We label states 
$b_{j_0}, b_{j_1}, \ldots, b_{j_{\log(X)-1}}$ by atomic proposition $\negt_1$
according to the bitsequence of ${n_1}$, atomic proposition $\negt_2$ 
according to the bitsequence of ${n_2}$, and atomic proposition $\pos$ 
according to the bitsequence of ${p}$. We reserve values $0$ and $|X|-1$ for 
$\bot$ and $\top$, respectively (see Fig~\ref{fig:system-tree-ae-app}).

\end{itemize}

\item {\em (Uncontrollable transitions $\uncont$)} For each state $v_x$, we add 
an uncontrollable transition from $v_x$ to any state $b_{j_0}$, where 
propositional variable $x$ appears as a positive literal in clause $y_j$. These 
transitions ensure that if a variable that appears a positive literal in two 
or more clauses, all or none of the associated clauses are preserved.

\item {\em (Controllable transitions $\cont$)} We represent each Horn clause 
as a branch of the plant. That is, we include the following transitions:

\begin{itemize}
 \item We connect the states that represent bitstrings as follows:
$$\Big\{(b_{j_i}, b_{j_{i+1}}) \mid j \in [1,m] \, \wedge i \in [0, 
\log(X))\Big\}.$$

\item We also connect the initial state to each $v_x$ state, where $x$ is a 
propositional variable that appears as a positive literal in some Horn clause.

\item Finally, we add a self-loop to the end of each branch, that is,
$$\Big\{ (b_{j_{\log(X)-1}}, b_{j_{\log(X)-1}}) \mid j \in [1,m] \Big\}.$$
\end{itemize}

\end{itemize}

It is easy to see that the plant that represents the Horn clauses is a tree. It 
branches into the paths that represent the clauses 
(see Fig.~\ref{fig:system-tree-ae-app}).

\noindent \textbf{HyperLTL formula.} \ We interpret the synthesized plant 
as a solution to HORN-SAT assigning false to every variable $x$ that appears as 
a positive literal on some path but $v_x$ is {\em not} reachable from the 
initial state and true to every variable $x'$ that appears as a positive 
literal on some path but $v_{x'}$ {\em is} reachable from the initial state.
We define a HyperLTL formula that ensures that this valuation satisfies the 
clause set. Let
$$\varphi_{\mathsf{map}} = \varphi_\bot \wedge \varphi_\top \wedge \varphi_C$$
be a HyperLTL formula with the following conjuncts:

\begin{itemize}

\item Formula $\varphi_\top$ enforces that $\top$ is assigned to true. This is 
expressed by requiring that, on all traces, $\top$ does not appear as a 
positive literal. That is, 
$$\varphi_\top = \forall \pi_1. \F(\neg\pos_{\pi_1}).$$

\item Formula $\varphi_\bot$ stipulates that $\bot$ is assigned to 
false. This is expressed by requiring that there exists a trace where $\bot$ 
appears as a positive literal. That is,
$$\varphi_\bot = \exists \pi_2. \G(\neg \pos_{\pi_2}).$$

\item Formula $\varphi_C$ ensures that all clauses are satisfied. This is 
expressed as a forall-exists formula that requires, for every trace in the 
synthesized plant, that for one of the variables that appear as 
negative literals in the clause, there must exist a trace where the same 
variable appears as the positive literal. That is,
\begin{align*}
\varphi_C = \forall \pi_1.\exists \pi_2. 
\G\Big(& (\negt_{1_{\pi_1}} \Leftrightarrow \pos_{\pi_2}) \, \vee \\ 
& (\negt_{2_{\pi_1}} \Leftrightarrow \pos_{\pi_2}) \Big).
\end{align*}

\end{itemize}
Overall, $\varphi_{\mathsf{map}}$ needs only one universal and one existential 
quantifier and it is straightforward to see that the input HORN-SAT formula is 
satisfiable if and only if the answer to the controller synthesis problem is 
affirmative.
\end{proof}

\subsection{The Full Logic}

We now turn to full HyperLTL. We show that the controller synthesis 
problem is in \comp{NP}. \comp{NP}-hardness holds already for the fragment with two universal quantifiers.

\begin{theorem}
\label{thrm:sys-tree-aae-upper}
\CS{HyperLTL}{\mbox{tree}} is in \comp{NP} in the size of the plant.
\end{theorem}

\begin{proof}
We nondeterministically guess a solution $\krip'$ in polynomial time. Since determining whether or not $\krip' \models \varphi$ 
can be solved in logarithmic space~\cite{bf18}, the synthesis problem is in 
\comp{NP}.
\end{proof}

%
\begin{theorem}
\label{thrm:sys-tree-aa-lower}
\CS{AA-HyperLTL}{\mbox{tree}} is \comp{NP-hard} in the size of the 
plant.
\end{theorem}

\begin{proof}
We reduce the 3-SAT problem to the controller synthesis problem. The 3SAT
problem is as follows:

\begin{quote}
Let $\{x_1, x_2, \dots, x_n\}$ be a set of propositional variables. Given is a 
Boolean formula $y = y_1 \wedge y_2 \wedge \cdots \wedge y_m$, where each 
$y_j$, for $j \in [1, m]$, is a disjunction of exactly three literals. Is 
$y$ satisfiable? That is, does there exist an assignment of truth values to 
$x_1, x_2, \dots, x_n$, such that $y$ evaluates to true.
\end{quote}

We now present a mapping from an arbitrary instance of 3SAT to the synthesis 
problem of a tree-shaped plant and a HyperLTL formula of the form 
$\forall \pi. \forall \pi'.\psi$. Then, we show that the plant satisfies the 
HyperLTL formula if and only if the answer to the 3SAT problem is affirmative. 
Figure~\ref{fig:sys-tree-aa-app} shows an example.

\noindent \textbf{Plant $\krip = \ktuple$: } 

\begin{itemize}
\item {\em (Atomic propositions $\AP$)} We include two atomic propositions: 
$\pos$ and $\negt$ to mark the positive and negative literals in each clause. 
Thus,
$$\AP = \big\{\pos, \negt\}.$$

\item {\em (Set of states $\States$)}  We now identify the members of $\States$:

\begin{itemize}

\item First, we include an initial state $\state_\init$. Then, for each clause 
$y_j$, where $j \in [1, m]$, we include a state $r_j$.

\item Let $y_j = (l, l', l'')$ be a clause in the 3SAT formula. We include 
the following set of states:
$$\Big\{v_{j_i}, v'_{j_i}, v''_{j_i} \mid i \in [1,n]\Big\},$$
where $n$ is the number of propositional variables.
If $l = x_i$ is in $y_j$, then we label state $v_{j_i}$ with proposition 
$\pos$. If $l = \neg x_i$ in $y_j$, then we label state $v_{j_i}$ 
with proposition $\negt$. We analogously label $v'$ and $v''$ associated to 
$l'$ 
and $l''$, respectively.

\end{itemize}
Thus, we have
\begin{align*}
\States = & \Big\{r_j \mid j \in [1,m]\Big\} 
\; \cup \\
& \Big\{v_{j_i}, v'_{j_i}, v''_{j_i} \mid i \in [1,n] \ \wedge \ j \in 
[1,m] \Big\}.
\end{align*}

\item {\em (Uncontrollable transitions $\uncont$)} We include an {\em 
uncontrollable} transition $(\state_\init, r_{j})$, for each clause $y_j$ in the 
3SAT formula, where $j \in [1, m]$:
$$
\inp = \Big\{(\state_\init, r_{j}) \mid j \in [1, m]\Big\}.
$$
These uncontrollable transitions ensure that during synthesis, all clauses are 
preserved.

\item {\em (Controllable transitions $\cont$)} We now identify the members of 
$\cont$:
 
\begin{itemize}


\item We connect all the states corresponding to the states corresponding 
to the propositional variables in a sequence. That is, we include the following 
transitions for each $j \in [1,m]$:
\begin{align*}
& \Big\{(r_{j}, v_{j_1}), (r_{j}, v'_{j_1}), 
(r_{j}, v''_{j_1})\Big\} \; \cup \\
& \Big\{(v_{j_i}, v_{j_{i+1}}), (v'_{j_i}, v'_{j_{i+1}}), (v''_{j_i}, 
v''_{j_{i+1}}) \mid i \in [1, n)\Big\}
\end{align*}

\item We also add a self-loop at each leaf state:
\begin{align*}
\Big\{(v_{j_n}, v_{j_n}), (v'_{j_n}, v'_{j_n}), (v''_{j_n}, v''_{j_n}) \mid  j 
\in [1,m] \Big\}
\end{align*}

\end{itemize}

\begin{figure}[t]
\centering
\scalebox{.8}{
\begin{tikzpicture}
 
 \tikzset{
    >=stealth,
    possible world/.style={circle,draw,thick,align=center},
    real world/.style={double,circle,draw,thick,align=center},
    minimum size=26pt
}

\coordinate (init) at (0, 0);

\node[draw,circle,text width=0.5cm,fill=white] (initstate) at ($ (init) + (3, 
1) $) {};

\draw[->] ($ (initstate) + (-.8,.5) $ ) -- (initstate);


\foreach \i in {0, 1} {
\node[draw,circle, minimum size=9mm ,fill=black] (r\i) at ($ (init) + 
(\i*6,-1.5) $) [text=white]
{
};

\pgfmathsetmacro\ii{int(\i+1)};
\node (a) at (r\i.west)[xshift=-3mm] {$r_{\ii}$};
}




\draw[->, dashed] (initstate) -- (r0);
\draw[->, dashed] (initstate) -- (r1);


\foreach \i in {0,1}{
\foreach \k in {0, 1, 2}{
\foreach \j in {1, 2, 3, 4}{
\ifthenelse{\(\i=0 \AND \k=1\) \OR \(\i=1 \AND \k=0\) \OR \(\i=1 \AND \k=2\)} 
{\def\clr{white}} {\def\clr{black!15}};
\draw [fill=\clr] ($ (init) + 
(r0)+ (0, -.8) + ( {(\i*6)+(\k*1.7-1.7)},\j*-1.7) $) ellipse (.65cm and .45cm) 
node 
(v\i\k\j)
{\ifthenelse{\(\i=0 \AND \j=3 \AND \k=2\) \OR \(\i=1 \AND \j=1 \AND \k=0\) \OR 
\(\i=1 \AND \j=2 \AND \k=1\)}{$\{\pos\}$}
{\ifthenelse{\(\i=0 \AND \j=1 \AND \k=0\) \OR \(\i=0 \AND \j=2 \AND \k=1\) \OR 
\(\i=1 \AND \j=4 \AND \k=2\) }{$\{\negt\}$}{}}{}};

\pgfmathsetmacro\ii{int(\i+1)};
\node (b) at (v\i\k\j)[xshift=-6mm,yshift=-5mm] 
{\ifthenelse{\k=0}{$v_{\ii_\j}$}{\ifthenelse{\k=1}{$v'_{\ii_\j}$}{$v''_{\ii_\j}$
} }};

}
\draw [->] (r\i) -- (v\i\k1.north);
}

}


\foreach \i in {0,1}{
\foreach \k in {0, 1, 2}{
\foreach \j/\n in {1/2,2/3,3/4}
{
\draw [->] (v\i\k\j) -- (v\i\k\n);
}}}


\foreach \i in {0,1} {
\foreach \k in {0,1,2} {
\path (v\i\k4) edge [loop below] (v\i\k4);
}}

\end{tikzpicture}
}
\caption{The plant for the 3SAT formula $(\neg x_1 \vee \neg x_2 
\vee x_3) \ \wedge \ (x_1 \vee x_2 \vee \neg x_4)$. The truth assignment $x_1 = 
\tru$, $x_2 = \false$, $x_3 = \false$, $x_4 = \false$ renders the tree with 
white branches, i.e., the grey branches are removed during synthesis.}
\label{fig:sys-tree-aa-app}
\end{figure}
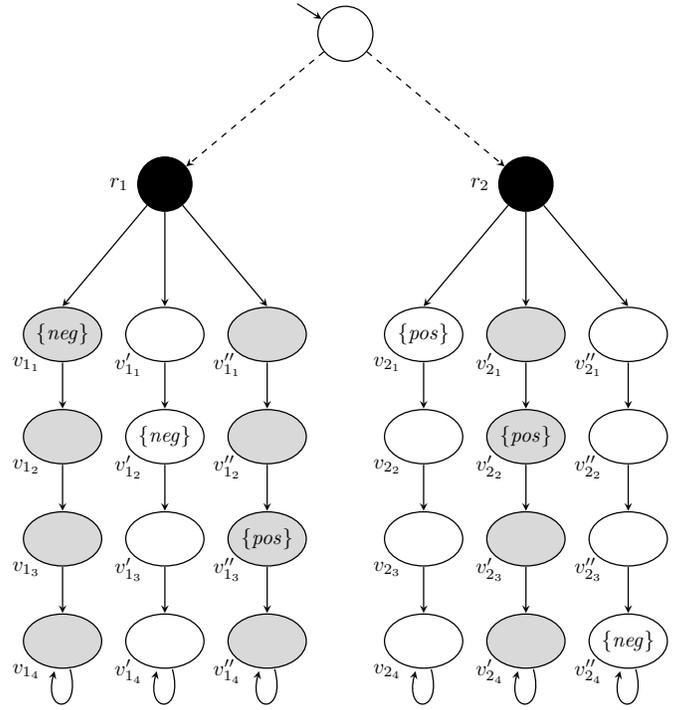

\end{itemize}

\noindent \textbf{HyperLTL formula: } The HyperLTL formula in our mapping
is the following:
\begin{align*}
\varphi_{\mathsf{map}} = & \forall \pi_1.\forall \pi_2. \G 
\Big(\neg \pos_{\pi_1} \, \vee \, \neg \negt_{\pi_2}\Big)
\end{align*}

We now show that the given 3SAT formula is satisfiable if and only if the 
plant obtained by our mapping can be controlled to satisfy the HyperLTL formula 
$\varphi_{\mathsf{map}}$:

\begin{itemize}
 \item ($\Rightarrow$) Suppose the 3SAT formula $y$ is satisfiable, i.e., 
there exists an assignment to the propositional variables $x_1, x_2, \dots, 
x_n$ that satisfies $y$. This implies that each $y_j$ becomes true, which 
in turn means that there exists at least one literal in each $y_j$ that 
evaluates to true. Now, given this assignment, we identify a plant 
$\krip'$ that satisfies the conditions of the controller synthesis problem stated in 
Section~\ref{sec:problem}. Suppose $y_j = (l \vee \l' \vee l'')$, for some $j 
\in [1, m]$. Also, suppose that $l = x_i$, for some $i \in [1,n]$. If 
$x_i = \tru$ in the answer to the 3SAT problem, then we keep states $v_{j_1}, 
v_{j_2}, \dots, v_{j_n}$ and all incoming and outgoing transitions to them. We 
also remove states $v'_{j_1}, v'_{j_2}, \dots, v'_{j_n}$ and $v''_{j_1}, 
v''_{j_2}, \dots, v''_{j_n}$. Likewise, suppose that $l = \neg x_i$, for some 
$i \in [1,n]$. If $x_i = \false$ in the answer to the 3SAT problem, then we 
keep states $v_{j_1}, v_{j_2}, \dots, v_{j_n}$ and all incoming and outgoing 
transitions to them. We also remove states $v'_{j_1}, v'_{j_2}, \dots, 
v'_{j_n}$ and $v''_{j_1}, v''_{j_2}, \dots, v''_{j_n}$ and all incoming and 
outgoing transitions to them. The case for literals $l', l''$ and states $v', 
v''$ follows trivially.

It is straightforward to see that the plant obtained after removing 
the aforementioned states satisfies $\varphi_{\mathsf{map}}$. This is 
because a propositional variable cannot be simultaneously true and false 
in the answer to the 3SAT problem. Thus, if we keep the states corresponding to 
a true variable $x_i$, the branches where some $v_{j_i}$ is labeled by 
$\negt$ is removed. The same argument holds for states $v'$ and $v''$ and the case 
where a variable evaluates to false.
 
\item ($\Leftarrow$) Suppose the answer to the synthesis problem is 
affirmative, i.e., there is a synthesized plant $\krip'$ 
that satisfies the HyperLTL formula $\varphi_{\mathsf{map}}$. This means that 
(1) all the $r$ states are preserved, since we cannot remove uncontrollable 
transitions, and (2) there are no pairs of $v$, $v'$, or $v''$ states at the 
same height of the tree of the synthesized plant, such that both $\pos$ 
and $\negt$ are true. We now describe how one can obtain a truth assignment to 
the propositional variables that satisfies the input 3SAT formula. Suppose that 
state $v_{j_1}, v_{j_2}, \dots, v_{j_n}$ appear in $\krip'$ for some $i$ and 
$j$, such that some state $v_{j_i}$ is labeled by $\pos$. We assign truth value 
$\tru$ to variable $x_i$. This assignment makes clause $y_j$ true, since $x_i$ 
is a literal in $y_j$. On the contrary, if state $v_{j_i}$ is labeled by 
$\negt$, then we assign truth value $\false$ to variable $x_i$. This 
assignment makes clause $y_j$ true, since $\neg x_i$ is a literal in $y_j$. 
Same argument holds for states $v'$ and $v''$. Furthermore, since all $r$ 
states are preserved all clauses evaluate to true and, hence, $y$ evaluates to 
$\tru$.
\end{itemize}
This concludes the proof.
\end{proof}

\begin{corollary}
\label{cor:sys-tree-aa}
 \CS{AA-HyperLTL}{\mbox{tree}} is \comp{NP-complete} in the size of the plant.
 
\end{corollary}

Corollary~\ref{cor:sys-tree-aa} shows a significant difference between program repair~\cite{bf19} and controller 
synthesis: while the program repair problem for the 
\comp{AA$^*$} fragment is \comp{L-complete}, the problem becomes 
\comp{NP-complete} for controller synthesis.

\begin{corollary}
\label{cor:sys-tree-hltl}
 \CS{HyperLTL}{\mbox{tree}} is \comp{NP-complete} in the size of the plant.
 
\end{corollary}

\section{Complexity of Controller Synthesis for Acyclic Graphs}
\label{sec:acyclic}

In this section, we analyze the complexity of the controller synthesis problem 
for acyclic plants.

\subsection{The Alternation-free Fragment}

We start with the existential fragment. It turns out that for this fragment, controller synthesis and model checking are actually the same problem; hence, the complexity of controller synthesis is \comp{NL-complete}, as known from model checking~\cite{bf18}.

\begin{theorem}
  \label{thrm:sys-acyc-e}
\CS{E$^*$-HyperLTL}{\mbox{acyclic}} is \comp{NL-complete} in the size of the 
plant.
\end{theorem}

\begin{proof}
  For existential formulas, the synthesis problem is equivalent to the model 
checking problem. A given plant satisfies the formula iff there is a solution 
to the synthesis problem, as we are only dealing with existential quantifiers. 
If the formula is satisfied, the witness to the synthesis problem is simply the 
original plant. Since the model checking problem for existential formulas over 
acyclic graphs is \comp{NL}-complete~\cite[Theorem 2]{bf18}, the same holds 
for the synthesis problem.
\end{proof}

We now switch to the universal fragment, where the complexity of the problem 
jumps to \comp{NP-complete}.

\begin{theorem}
  \label{thrm:sys-acyc-a}
\CS{A$^*$-HyperLTL}{\mbox{acyclic}} is \comp{NP-complete} in the size of the 
plant.
\end{theorem}

\begin{proof}
For the upper bound, one can guess a solution to the synthesis problem and 
verify in polynomial time. The lower bound follows the lower bound of 
\CS{AA-HyperLTL}{\mbox{tree}} shown in Theorem~\ref{thrm:sys-tree-aa-lower}.
\end{proof}


\subsection{Formulas with Quantifier Alternation}

We first build on Theorem~\ref{thrm:sys-acyc-a} to study the complexity of the 
synthesis problem for the \comp{E$^*$A$^*$} Fragment.

\begin{theorem}
  \label{thrm:sys-acyc-ea}
\CS{E$^*$A$^*$-HyperLTL}{\mbox{acyclic}} is in $\mathsf{\Sigma_2^p}$ in the 
size of the plant.
\end{theorem}

\begin{proof}
We show membership to $\mathsf{\Sigma_2^p}$. Since the plant is acyclic, the 
length of the traces is bounded by the number of states. We can thus 
nondeterministically guess the witness to the existentially quantified traces 
in polynomial time, and then solve the problem for the remaining formula, 
which has only universal quantifiers. By Theorem~\ref{thrm:sys-acyc-a} (i.e., 
\comp{NP-hardness} of the problem for the \comp{A$^*$} fragment), it holds that 
\CS{E$^*$A$^*$-HyperLTL}{\mbox{acyclic}} is in $\mathsf{\Sigma_2^p}$.
\end{proof}

Next, we consider formulas where the number of quantifier alternations is 
bounded by a constant $k$. We show that changing the frame structure from 
trees to acyclic graphs results in a significant increase in complexity (see 
Table~\ref{tab:system}). The complexity of the synthesis problem is 
similar to the model checking problem, with the synthesis problem being one 
level higher in the polynomial hierarchy (cf.~\cite{bf18}). This also 
means that complexity of the problem is aligned with the complexity of the 
model repair problem~\cite{bf19}.

\begin{theorem} For $k\geq 2$,
\label{thrm:system-acyc-EAk1}
\CS{\mbox{(EA)}$k$\mbox{-HyperLTL}}{\mbox{acyclic}} is 
\comp{$\mathsf{\Sigma^p_{k}}$-complete} in the size of the plant.
 For $k \geq 1$,  \CS{\mbox{(AE)}$k$\mbox{-HyperLTL}}{\mbox{acyclic}} is
\comp{$\mathsf{\Sigma^p_{k+1}}$-complete} in the size of the plant.

\end{theorem}

\begin{proof}
We show membership in \comp{$\mathsf{\Sigma^p_{k}}$} and 
\comp{$\mathsf{\Sigma^p_{k+1}}$}, respectively, as follows. Suppose that the 
first quantifier is existential.
Since the plant is acyclic, the length of the traces is bounded by the number 
of 
states. We can thus nondeterministically guess the witness to the 
existentially quantified traces in polynomial time, and then verify the 
correctness of the guess by
model checking the remaining formula, which has $k-1$ quantifier alternations
and begins with a universal quantifier. The verification can be done 
in \comp{$\mathsf{\Pi^p_{k-1}}$}~\cite[Theorem 3]{bf18}. Hence, the synthesis 
problem is in \comp{$\mathsf{\Sigma^p_{k}}$}.

If the first quantifier is universal, we apply the same procedure except that 
we only guess the solution to the synthesis problem (there are no leading 
existential quantifiers). In this case, the formula for the model checking 
problem has $k$ quantifier alternations. Hence, we solve the model checking 
problem in \comp{$\mathsf{\Pi^p_{k}}$} and the synthesis problem in 
\comp{$\mathsf{\Sigma^p_{k+1}}$}.

We give a matching lower bound for 
\CS{\mbox{(AE)}$k$\mbox{-HyperLTL}}{\mbox{acyclic}}. Since the
$\mbox{(AE)}k\mbox{-HyperLTL}$  formulas are contained
in the $\mbox{(EA)}k+1\mbox{-HyperLTL}$ formulas (not using the outermost 
existential quantifiers), this also provides a matching lower bound for 
\CS{\mbox{(EA)}$k$\mbox{-HyperLTL}}{\mbox{acyclic}}.

We establish the lower bound for 
\CS{\mbox{(AE)}$k$\mbox{-HyperLTL}}{\mbox{acyclic}} via a reduction from the 
{\em 
quantified Boolean formula} (QBF) satisfiability problem~\cite{gj79}:

\begin{quote}

{\em Given is a set of Boolean variables, $\{x_1, x_2, \dots, x_n\}$, and a 
quantified Boolean formula
$$y=\quant_1 x_1.\quant_1 x_2\dots\quant_{n-1} x_{n-1}.\quant_n x_n.(y_1 \, 
\wedge \, y_2 \, \wedge \dots \wedge \, y_m)$$
where each $\quant_i \in \{\forall, \exists\}$ ($i \in [1, n]$) and each clause 
$y_j$ ($j \in [1, m]$) is a disjunction of three literals (3CNF). Is $y$ 
true?}
 
\end{quote}
If $\quant_1 = \exists$ and $y$ is restricted to at most $k$
alternations of quantifiers, then QBF satisfiability is complete for
\comp{$\mathsf{\Sigma^p_{k}}$}.
We note that in
the given instance of the QBF problem:

\begin{itemize}

\item The clauses may have more than three literals, but three is sufficient of 
our purpose;

\item The inner Boolean formula has to be in conjunctive normal form in order 
for our reduction to work;

\item Without loss of generality, the variables in the literals of the same 
clause are different (this can be achieved by a simple pre-processing of the 
formula), and

\item If the formula has $k$ alternations, then it has $k+1$ alternation {\em 
depths}. For example, formula
$$\forall x_1.\exists x_2. (x_1 \vee \neg x_2)$$ 
has one alternation, but two alternation depths: one for $\forall x_1$ and the 
second for $\exists x_2$. By $d(x_i)$, we mean the alternation depth of Boolean 
variable $x_i$.

\end{itemize}

We now present, for $k \geq 1$, a mapping from an arbitrary instance of QBF with 
$k$ alternations and where $\quant_1 = \exists$ to the synthesis problem of an 
acyclic plant and a HyperLTL formula with $k-1$ quantifier alternations and a 
leading universal quantifier. Then, we show that the plant can be pruned so 
that it satisfies the HyperLTL formula if and only if the answer to the QBF 
problem is affirmative.

The reduction is similar to the reduction from QBF satisfiability to the 
HyperLTL model checking problem~\cite[Theorem 3]{bf18} except for the treatment 
of the outermost existential quantifiers. In the reduction to the model 
checking problem, these quantifiers are translated to trace quantifiers, 
resulting in a HyperLTL formula with $k$ quantifier alternations and a leading 
existential quantifier. In the reduction to the synthesis problem, the 
outermost existential quantifiers are resolved by the pruning of controllable 
transitions. For this reason, it suffices to build a HyperLTL formula with one 
less quantifier alternation, i.e., with $k-1$ quantifier alternations, and a 
leading universal quantifier.

In the following, we first describe the plant from the reduction to the model 
checking problem~\cite[Theorem 3]{bf18} and then describe the necessary 
additions for the reduction to the synthesis problem. 
Figure~\ref{fig:system-acyclic-qbf} shows an example.

\begin{figure}[t]
\centering
\scalebox{.8}{
\pgfdeclarelayer{bg}    
\pgfsetlayers{bg,main}  

\begin{tikzpicture}

\coordinate (init) at (0, 0);

\node[draw,circle,text width=0.3cm,fill=white] (initstate) at ($ (init) + 
(3, 2) $) {};
\draw[->] ($ (initstate) + (-.5, .5) $) -- (initstate);

\node[draw,circle,text width=0.3cm,fill=black!60] (sh0) at (init) {};

\foreach \i/\j in {0/1,1/2, 2/3}
{

\draw [fill=white] ($ (sh\i) + (-1,-1) $) ellipse (.65cm and .3cm) node (s\i)
{$\{q^{\j}, p\}$};

\draw [fill=white]($ ({sh\i}) + (1,-1) $) ellipse (.65cm and .3cm) node (sb\i) 
{$\{q^{\j}, \bar{p}\}$};

\node[draw,circle,text width=0.25cm,fill=black!60] (sh\j) at ($ (sh\i) +
(0,-2) $)  {};

\draw [->, dashed] (sh\i) -- (s\i);
\draw [->, dashed] (sh\i) -- (sb\i);
\draw [->, dashed] (s\i) -- (sh\j);
\draw [->, dashed] (sb\i) -- (sh\j);

}

\path (sh3) edge [loop below] (1);

\foreach \i/\j in {0/1,1/2}
{

\node[draw,circle,text width=0.25cm,fill=black] (u\i0) at ($ (init) + 
(\j*3+.5,0) $) [text=white]{$\hspace*{-0.1cm}\{c\}$};

\foreach \x/\y in {0/1,1/2, 2/3}
{

\ifthenelse{\i = 0}
{
\ifthenelse{\x = 0}
{
\draw [fill=white]($ (u\i\x) + (0,-1) $) ellipse (.65cm and .3cm) node (v\i\x)
{$\{q^{\y}, p\}$};
}{}
\ifthenelse{\x = 1}
{
\draw [fill=white]($ (u\i\x) + (0,-1) $) ellipse (.65cm and .3cm) node (v\i\x)
{$\{q^{\y}, \bar{p}\}$};
}{}
\ifthenelse{\x = 2}
{
\draw [fill=white]($ (u\i\x) + (0,-1) $) ellipse (.65cm and .3cm) node (v\i\x)
{$\{q^{\y}, p\}$};
}{}
}{}

\ifthenelse{\i = 1}
{
\ifthenelse{\x = 0}
{
\draw [fill=white]($ (u\i\x) + (0,-1) $) ellipse (.65cm and .3cm) node (v\i\x)
{$\{q^{\y}, \bar{p}\}$};
}{}
\ifthenelse{\x = 1}
{
\draw [fill=white]($ (u\i\x) + (0,-1) $) ellipse (.65cm and .3cm) node (v\i\x)
{$\{q^{\y}, {p}\}$};
}{}
\ifthenelse{\x = 2}
{
\draw [fill=white]($ (u\i\x) + (0,-1) $) ellipse (.65cm and .3cm) node (v\i\x)
{$\{q^{\y}, \bar{p}\}$};
}{}
}{}

\node[draw,circle,text width=0.25cm,fill=black!60] (u\i\y) at ($ (v\i\x) +
(0,-1) $)  {};

\draw [->] (u\i\x) -- (v\i\x);
\draw [->] (v\i\x) -- (u\i\y);
}

\path (u\i3) edge [loop below] (\i);

}

\draw [->, dashed] (initstate) -- (sh0);
\draw [->, dashed] (initstate) -- (u00);
\draw [->, dashed] (initstate) -- (u10);

\begin{pgfonlayer}{bg}

\draw [rounded corners=10,,dotted,fill=black!15] ($ (init) + (-2,.5) $) 
rectangle ++(4,-7.5) node (r1) {};
\node (text1) at ($ (init) +  (-1,.8) $) {$\pi_d$ traces};

\draw [rounded corners=10,,dotted,fill=black!5] ($ (init) + (2.5,.5) $) 
rectangle ++(5,-7.5) node (r2) {};
\node (text1) at ($ (init) +  (7,.8) $) {$\pi'$ traces};

\end{pgfonlayer}

\end{tikzpicture}
}
\caption{Plant for the formula $y = \exists x_1.\forall 
x_2.\exists x_3.(x_1 \vee \neg x_2 \vee x_3) \wedge 
(\neg x_1 \vee x_2 \vee \neg x_3)$.}
\label{fig:system-acyclic-qbf}
\end{figure}
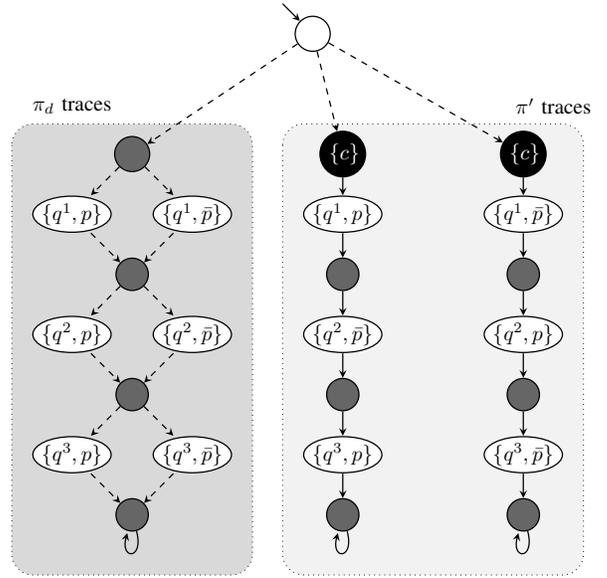

\noindent \textbf{Plant $\krip = \ktuple$: } 

\begin{itemize}

\item {\em (Atomic propositions $\AP$)} For each alternation depth $d \in [1, 
k+1]$, we include an atomic proposition $q^d$. We furthermore include three 
atomic propositions: $c$ is used to mark the clauses, $p$ is used to force 
clauses to become true if a Boolean variable appears in a clause, and 
proposition $\bar{p}$ is used to force clauses to become true if the negation 
of a Boolean variable appears in a clause in our reduction. 

\item {\em (Set of states $\States$)}  We now identify the members of $\States$:

\begin{itemize}

\item First, we include an initial state $\state_\init$ and a state $r_0$. 
Then, for each clause $y_j$, where $j \in [1, m]$, we include a state 
$r_j$, labeled by proposition $c$.
  
\item For each clause $y_j$, where $j \in [1, m]$, we introduce the 
following $2n$ states: 
$$\Big\{v^j_i, u^j_i \mid i \in [1, n]\Big\}.$$
Each state $v^j_i$ is labeled with propositions $q^{d(x_i)}$, and with $p$ if 
$x_i$ is a literal in $y_j$, or with $\bar{p}$ if $\neg x_i$ is a literal in 
$y_j$.

\item For each Boolean variable $x_i$, where $i \in [1, n]$, we include three 
states $s_i$, $\bar{s}_i$, and $\hat{s}_i$. Each state $s_i$ (respectively, 
$\bar{s}_i$) is labeled by $p$ and $q^{d(x_i)}$ (respectively, $\bar{p}$ and
$q^{d(x_i)}$).

\end{itemize}
Thus,
\begin{align*}
S = & \big\{s_\init \big\} \, \cup \, \big\{r_j \mid j \in [0, m]\big\}  \; 
\cup \\
& \big\{v^j_i, u^j_i, s_i, \bar{s_i}, \hat{s}_i \mid i \in [1, n] \wedge 
j \in [1,m]\big\}.
\end{align*}

\item {\em (Uncontrollable transitions $\uncont$)} The set of uncontrollable 
transitions include the following:

\begin{itemize}

\item We add outgoing transitions from the initial state of $\krip$ to clause 
states as well as the state that represent the first propositional variable.
The idea here is to make these transitions uncontrollable to ensure that during 
synthesis the clauses and the diamond structure for all alternation depths 
except $k + 1$ are preserved.

\item For each $i \in [1, n]$, we include transitions $(s_i, \hat{s}_i)$ and 
$(\bar{s}_i, \hat{s}_i)$. For each $i \in [1, n)$, we include transitions 
$(\hat{s}_i, s_{i+1})$ and $(\hat{s}_i, \bar{s}_{i+1})$.

\end{itemize}

Thus,
\begin{align*}
\uncont = & \big\{(\state_\init, r_j) \mid j \in [0, m]\big\} \; \cup
\\
& \big\{(r_0, s_1), (r_0, \bar{s}_1) \big\} \; \cup\\
& \big\{(s_i, \hat{s}_i), (\bar{s}_i, \hat{s}_i) \mid i \in [1, n] \big\} \; 
\cup\\
& \big\{(\hat{s}_i, s_{i+1}), (\hat{s}_i, \bar{s}_{i+1}) \mid i \in [1, n) 
\big\}.
\end{align*}

\item {\em (Controllable transitions $\cont$)} We now identify the 
members of $\cont$:
 
\begin{itemize}
 
\item We add transitions $(r_j, v^j_1)$ for each $j \in [1, m]$.

\item For each $i \in [1, n]$ and $j \in [1, m]$, we include transitions 
$(v^j_i, u^j_i)$. For each $i \in [1, n)$ and $j \in [1, m]$, we include 
transitions $(u^j_i, v^j_{i+1})$.

\item We include two transitions $(r_0, s_1)$ and $(r_0, \bar{s}_1)$.

\item Finally, we include self-loops $(\hat{s}_n, \hat{s}_n)$ and $(u_n^j, 
u_n^j)$, for each $j \in [1, m]$.

\end{itemize}
Thus,
\begin{align*}
\cont = & \big\{(r_j, v^j_1), (u_n^j, u_n^j) \mid j \in [1, m] \big\} \; \cup\\
& \big\{(v^j_i, u^j_i) \mid i \in [1, n] \, \wedge \, j \in [1, m] \big\} \; 
\cup\\
& \big\{(u^j_i, v^j_{i+1}) \mid  i \in [1, n) \, \wedge \, j \in [1, m] \big\}.
\end{align*}

\end{itemize}



\newcommand{\map}{\mathsf{map}}

\noindent \textbf{HyperLTL formula: }  The role of the HyperLTL formula in our 
reduction is to ensure that the QBF instance is satisfiable iff the HyperLTL 
formula is satisfied on the solution to the synthesis problem:
\[ \begin{array}{l}
\label{eq:hltlsys}
\varphi_{\map} = \forall \pi_{k+1}.\forall \pi_{k}.\exists \pi_{k-1} \cdots 
\exists \pi_2. 
\forall \pi_1. \forall \pi'. \\
\qquad  \Bigg( \bigwedge_{d \in \{1, 3, \dots, k\}} \X 
\neg c_{\pi_d} \, \wedge \, \X c_{\pi'}\Bigg) \; \Rightarrow \\
\nonumber \qquad \Bigg(\bigwedge_{d \in \{2, 4, \dots, k+1\}} \X \neg c_{\pi_d} 
 \; \wedge \\
\qquad ~~~~\F\bigg[\bigvee_{d \in [1,k+1]} \Big(\big(q^{d}_{\pi_d} 
\Leftrightarrow q_{\pi'}^d\big) \, \wedge \\
\qquad~~~~~~~~~~~~~~~ \big((p_{\pi'} \wedge p_{\pi_d}) \; 
\vee \; (\bar{p}_{\pi'} \wedge \bar{p}_{\pi_d})\big)\Big) 
\bigg]\Bigg)
\end{array}\]

Intuitively, $\varphi_\map$
expresses the following: for all the clause traces $\pi'$ and all traces that 
valuate universally quantified variables ($\pi_1, \pi_3, \ldots$), there exist 
traces evaluating the existentially quantified variables ($\pi_2, \pi_4, 
\ldots$), where either $p$ or $\bar{p}$ eventually matches its counterpart 
position in the clause trace $\pi'$. The dependencies between the trace 
quantifiers for the valuation of the variables match the dependencies in the 
quantified
Boolean formula. A special case are the outermost existential variables in the 
QBF. The corresponding trace quantifier (for $\pi_{k+1}$) is universal, rather 
than existential. As a result, the formula has only $k-1$ alternations.
We allow synthesis to reduce the valuations for the outermost existential 
variables to a single valuation. Hence, universal and existential 
quantification 
is the same. 
\end{proof}

\medskip

Finally, Theorem~\ref{thrm:system-acyc-EAk1} implies that the synthesis 
problem for acyclic plants and HyperLTL formulas with
an arbitrary number of quantifiers is in \comp{PSPACE}.

\begin{corollary}
\label{cor:sys-acyclic-hltl}
\CS{{HyperLTL}}{\mbox{acyclic}} is in \comp{PSPACE} in the size of the 
plant.  
 
\end{corollary}

\section{Complexity of Controller Synthesis for General Graphs}
\label{sec:general}

In this section, we investigate the complexity of the controller synthesis 
problem for general graphs. We again begin with the alternation-free fragment 
and then continue with formulas with quantifier alternation.

\subsection{The Alternation-free Fragment}

We start with the existential fragment. As for acyclic graphs, 
the controller synthesis problem is already solved by model checking.

\begin{theorem}
\label{thm:sys-general-e}
\CS{$\mbox{E}^*$-HyperLTL}{\mbox{general}} is \comp{NL-complete} in the size of 
the plant.
\end{theorem}

\begin{proof}
Analogously to the proof of Theorem~\ref{thrm:sys-acyc-e}, we note that, for 
existential formulas, the synthesis problem is equivalent to the model checking 
problem. A given plant satisfies the formula if and only if it has a 
solution to the synthesis problem. If the formula is satisfied, then the 
solution is simply the original plant. Since the model checking problem for 
existential formulas for general graphs is \comp{NL}-complete~\cite{frs15}, the 
same holds for the synthesis problem.
\end{proof}

Similar to tree-shaped and acyclic graphs, the synthesis problem for the 
universal fragment is also \comp{NP-complete}. 

\begin{theorem}
\CS{$\mbox{A}^+$-HyperLTL}{\mbox{general}} is \comp{NP-complete} in the size 
of the plant.
\label{thm:sys-general-a}
\end{theorem}

\begin{proof} For membership in \comp{NP}, we nondeterministically guess a 
solution to the synthesis problem, and verify the correctness of the 
universally quantified HyperLTL formula against the solution in polynomial time 
in the size of the plant. \comp{NP-hardness} follows from the 
\comp{NP-hardness} of the synthesis problem for LTL~\cite{bek09}.
%
%
\end{proof}

\subsection{Formulas with Quantifier Alternation}

Next, we consider formulas where the number of quantifier
alternations is 
bounded by a constant $k$. We show that changing
the frame structure from 
acyclic to general graphs results in a
significant increase in complexity (see 
Table~\ref{tab:system}).

\begin{theorem}
  \label{thrm:system-general-EAk}
\CS{$\mbox{E}^*\mbox{A}^*\mbox{-HyperLTL}$}{\mbox{general}} is in
\comp{PSPACE} in the size of the plant.
\CS{$\mbox{A}^*\mbox{E}^*\mbox{-HyperLTL}$}{\mbox{general}} is
\comp{PSPACE}-complete in the size of the plant.
  For $k \geq 2$,
\CS{$\mbox{(EA)}^k\mbox{-HyperLTL}$}{\mbox{general}} and
\CS{$\mbox{(AE)}^k\mbox{-HyperLTL}$}{\mbox{general}}
are
\comp{$(k{-}1)$-EXPSPACE}-complete in the size of the plant.
\end{theorem}
 
\begin{proof}
The claimed complexities are those of the model checking 
problem~\cite{markus}. For the upper bound, we guess, in \comp{PSPACE},
a solution to the control problem and then verify, using the
model checking algorithm for the considered fragment of HyperLTL, that the
solution satisfies the HyperLTL formula.

For the lower bound, we reduce the model checking problem to the
controller synthesis problem by identifying each transition of the
given plant as an uncontrollable transition of the plant. In this 
way, the controller synthesis cannot modify the plant, and the synthesis 
succeeds iff the given plant in the model checking problem already 
satisfies the HyperLTL formula.
\newcommand\toberemoved[1]{}
\toberemoved{
  We prove that the synthesis problem has the same complexity as the model 
checking problem.
  For the upper bound, we enumerate, in \comp{PSPACE}, all possible 
solutions to the synthesis problem, and then verify against the HyperLTL 
formula.

  For the lower bound, we modify the plant and the HyperLTL formula 
such that the only possible solution to synthesis is the unchanged plant.
  After the modification, the repair problem thus has the same result as 
the 
model checking problem.
  The idea of the modification is to assign numbers to the successors of each 
state. We add extra states such that the
  traces that originate from these states correspond to all possible number 
sequences. Finally, the HyperLTL formula states
  that for each such number sequence there exists a corresponding trace in the 
original plant.
  A technical difficulty is that the HyperLTL formula needs to be fixed for all 
plants, but different plants may differ in their branching degree. For 
this reason we first 
(Step 1) transform the given Kripke structure and HyperLTL
  formula into an equivalent problem where every state has precisely two 
successors; afterwards (Step 2) we carry out
  the modification that ensures that the repair cannot remove any edges.

  \emph{Step~1: Fixing the branching degree.} We first translate the given 
Kripke structure $\krip = \ktuple$ into a Kripke structure
  $\krip'=\langle S', s_\init', \trans', L' \rangle$ where every state has at 
most two successors; afterwards we construct another Kripke structure $\krip''$
  where every state has exactly two successors. Let $d(s) = | \{s' \in S \mid 
(s,s') \in \trans\}$ be
  the branching degree of state $s \in S$ and let $d = \max_{s\in S} d(s)$ be 
the branching degree. We add a marker $m$ to identify the states of the 
original 
Kripke structure. \\

  \begin{itemize}
  \item $\AP' = \AP \cup \{m\}$.
  \item $\States' = S \cup \{ (s,i) \mid s \in \States, 2 \leq i \leq d(s) \}$
  \item $s_\init' = s_\init$
  \item $\trans' = \bigcup_{s\in S}\{ (s,s_1'), (s,(s,2)),
  ((s,2),s_2'),((s,2),s_2'),$\\ \mbox{\qquad} $((s,2),(s,3)), ((s,3), s_3'), 
\ldots,$\\ \mbox{\qquad} $((s,d(s)-1), (s,s_{d(s)-1})), (s,d(s)-1), s_{d(s)}') 
\}$
  \item  $L'(s) = L'(s) \cup \{m\}$ for $s \in \States$ and\\ $L(s,i)= 
\emptyset$ for $(s,i) \in \States'$
  \end{itemize}
  
Let $\varphi$ be a given HyperLTL formula over $\AP$. We define a new HyperLTL 
formula $\varphi'$ over $\AP'$ such that $\krip$ satisfies $\varphi$ iff 
$\krip'$ satisfies $\varphi'$.
We transform $\varphi$ inductively as follows:
\begin{itemize}
\item if $\varphi = \X \psi$, then $\varphi' = (\neg m)\, \mathcal U\, (m 
\wedge 
\psi')$, 
\item if $\varphi = \psi_1\, \mathcal U\, \psi_2$, then $\varphi' = (m 
\rightarrow \psi')\, \mathcal U\, (m \wedge \psi_2)$,
\item if $\varphi = \exists \pi.\ \psi$ then $\varphi' = \exists \pi.\ \psi'$
\item if $\varphi = \forall \pi.\ \psi$ then $\varphi' = \forall \pi.\ \psi'$
\item if $\varphi = \neg \psi$ then $\varphi' = \neg \psi'$
\item if $\varphi = \psi_1 \vee \psi_2$ then $\varphi' = \psi_1' \vee \psi_2'$
\end{itemize}
where $\psi',\psi_1',\psi_2'$ are the transformations of $\psi$, $\psi_1$, and 
$\psi_2$, respectively.

The Kripke structure $\krip'$ may still contain states that have only one 
successor. We construct a Kripke structure $\krip''=\langle S'', s_\init'', 
\trans'', L'' \rangle$ that has the same traces over $\AP'$ as $\krip'$, and, 
hence, satisfies the same HyperLTL formulas over $\AP'$, but only has states 
with exactly two successors. We furthermore distinguish the two successors with 
a fresh proposition $c$. 
\begin{itemize}
\item $\AP'' = \AP' \cup \{c\}$. 
\item $\States'' = (S' \times \{0,1\})$
\item $s_\init'' = (s_\init',0)$
\item $\trans'' = \bigcup_{(s,i)\in S''}\{ ((s,i), (s',0)), ((s,i),(s'',1)) 
\mid 
(s,s'), (s,s'') \in \trans$ and $s'<s''$ if $d(s)=2$ and $s'=s''$ if $d(s)=1 
\}$ 
 
\item $L''(s,0) = L'(s)$ and\\ $L''(s,1) = L'(s) \cup \{c\}$ for $s \in 
\States'$
\end{itemize}
where $<$ is an arbitrary order on the states in $\States'$.

\emph{Step~2: Protecting the Kripke structure from repair.} As described 
above, we construct another Kripke structure $\krip''=\langle S'', s_\init''\, 
\trans''\, L''\ \rangle$ by adding fresh states that generate all sequences of 
successor numbers, i.e., all bitsequences of $c$. We again use a marker 
proposition $m''$ to mark the states from $S''$.
 \begin{itemize}
  \item $\AP''' = \AP'' \cup \{m''\}$.
  \item $\States''' = S'' \cup \{ 0,1 \}$
  \item $s_\init''' = s_\init''$
  \item $\trans''' = \trans'' \cup \{ (s_\init''', 0), (s_\init''', 1), (0,1), 
(0,0), (1,0), 1,1) \}$
  \item $L'''(s) = L''(s) \cup \{m''\}$ for $s \in \States''$, 
$L(0)=\emptyset$, 
and $L(1) = \{c\}$.
  \end{itemize}
 Finally, we modify the HyperLTL formula $\varphi'$ so that $\varphi'$ only 
refers to traces in $S''$ and that no transitions can be removed by the 
repair.
 Let $\varphi' = Q_1 \pi_1 Q_2 \pi_2 \ldots Q_m \pi_m .\ \psi$, where $\psi$ is 
quantifier-free. Let $E$ be the set of indices $i$ such that $Q_i$ is 
existential, and
 $A$ be the set of indices $i$ such that $Q_i$ is universal.
 We define \[
 \varphi''' = Q_1 \pi_1 Q_2 \pi_2 \ldots Q_m \pi_m .\ (\bigwedge_{m \in E} \X 
m''_{\pi_i}) \wedge ( (\bigwedge_{m \in E} \X m''_{\pi_i}) \rightarrow \psi)
 \]
 The final HyperLTL formula $\varphi''''=\varphi''' \wedge \varphi_{01} \wedge 
\varphi_{\krip}$ is the conjunction of $\varphi'''$ and two constraints that 
protect the Kripke structure from repair.
 \begin{itemize}
 \item $\varphi_{01}$ protects the transitions to states $0$ and $1$:
   \[ \begin{array}{l}
     \varphi_{01} = \exists \pi . \exists \pi' .  \exists \pi''  . \exists 
\pi''' .\\
     \qquad (\X (\neg m_\pi'' \wedge \neg c \wedge \X \neg c)) \wedge\\
     \qquad (\X (\neg m_{\pi'}'' \wedge \neg c \wedge \X c)) \wedge\\
     \qquad (\X (\neg m_{\pi''}'' \wedge c \wedge \X  c)) \wedge\\
     \qquad (\X (\neg m_{\pi'''}'' \wedge c \wedge \X \neg c))
     \end{array}\]

\item $\varphi_{\krip}$ protects all transitions in $\krip''$:
  \[ \begin{array}{l} \varphi_{\krip} = \forall \pi \exists \pi'. \ (\X\G \neg 
m''_\pi) \rightarrow\\
    \qquad \X\G (m''_{\pi'} \wedge c_\pi \leftrightarrow c_{\pi'})
    \end{array}\] 
    \end{itemize}
 
 } 
\end{proof}

\medskip

Finally, Theorem~\ref{thrm:system-general-EAk} implies that the repair 
problem
for general plants and HyperLTL formulas with
an arbitrary 
number of quantifiers is in \comp{NONELEMENTARY}.

\begin{corollary}
\label{cor:sys-general-hltl}
\CS{{HyperLTL}}{\mbox{general}} is \comp{NONELEMENTARY} in the size of the 
plant. 
\end{corollary}

\section{Related Work}
\label{sec:related}


There has been a lot of recent progress in automatically
{verifying}~\cite{frs15,fmsz17,fht18,cfst19}
and {monitoring}~\cite{ab16,fhst19,bsb17,bss18,fhst18,sssb19,hst19} 
HyperLTL specifications. HyperLTL is also supported by a growing set of 
tools, including the model checker MCHyper~\cite{frs15,cfst19}, the  
satisfiability checkers EAHyper~\cite{fhs17} and MGHyper~\cite{fhh18}, and the 
runtime monitoring tool RVHyper~\cite{fhst18}.

The closest work to the study in this paper is the analysis of the  
{\em program repair} problem for HyperLTL~\cite{bf19}. The repair problem is to find a 
subset of traces of a Kripke structure that satisfies a given HyperLTL formula. 
Thus, the repair problem is similar to the controller synthesis problem studied 
in this paper. In both problems, the goal is to prune the set of transitions of the 
given plant or model. However, in program repair, all transitions are 
controllable, whereas in controller synthesis the pruning cannot be applied to 
uncontrollable transitions. We draw the following comparison and contrast 
between the results in~\cite{bf19} and this paper:

\begin{itemize}

\item The general positive result of this paper is that although 
controller synthesis is typically perceived as a more difficult problem than 
program repair (due to the existence of uncontrollable transitions), our study, 
summarized in Table~\ref{tab:system}, shows that the complexity of controller 
synthesis and program repair for HyperLTL remain pretty close. This result 
may appear counterintuitive. For some cases, there is a simple explanation
why the complexities are similar. For example, for general graphs, the 
complexity is dominated by the model checking complexity. Both problems can be 
solved by first guessing a solution and then verifying it. The complexity of 
the model checking problem (which is the dominating factor) is the same for 
both problems and as a result, the complexity of program repair and controller 
synthesis is the same (\comp{NONELEMENTARY}).
However, this is not always the case. 
For example, for acyclic graphs, the fact that guessing the controller is more difficult
than guessing the repair leads to a higher complexity (within the polynomial hierarchy).
Another interesting observation is the effect of the universal quantifiers in the $\forall\forall$ fragment.  
While the repair problem for the $\forall\forall$ fragment is 
\comp{L-complete} for tree-shaped graphs, the problem becomes \comp{NP-complete} for the 
controller synthesis problem. Also, while the repair problem for the 
$\forall\forall$ fragment is \comp{NL-complete} for acyclic graphs, 
it becomes \comp{NP-complete} for the controller synthesis problem. This is 
significant, because many important security policies such as certain types of
noninterference~\cite{gm82} and observational determinism~\cite{zm03} fit in 
this fragment.

\item Although the proof techniques in this paper are similar to those in~\cite{bf19}, leveraging the existence of uncontrollable 
transitions has made our lower bound proofs more elegant. In particular, the 
proofs in~\cite{bf19} need to incorporate complex constraints in the HyperLTL 
formulas 
to make sure that during reductions, parts of the state space related to 
clauses of the input (e.g., SAT or QBF) formulas are not removed. Here, we 
mimic this in a more elegant way by using uncontrollable transitions, which cannot be removed during synthesis.

\end{itemize}

The controller synthesis problem studied in this paper is also related to classic \emph{supervisory control}, where, 
for a given plant, a supervisor is constructed that selects an appropriate 
subset of the plant's controllable actions to ensure that the resulting 
behavior is safe~\cite{tw87,l91,jk06}.

Directly related to the controller synthesis problem studied in this paper is 
the {\em satisfiability}. The satisfiability problem for HyperLTL was shown to 
be decidable for the $\exists^*\forall^*$ fragment and for any fragment that 
includes a $\forall\exists$ quantifier alternation~\cite{fh16}. The hierarchy of 
hyperlogics beyond HyperLTL has been studied in~\cite{cfhh19}.

The general \emph{synthesis} problem differs from controller synthesis in that the solutions are not limited to the state graph of the plant.
For HyperLTL, synthesis was shown to be undecidable in 
general, and decidable for the $\exists^*$ and $\exists^*\forall$ fragments~\cite{DBLP:journals/acta/FinkbeinerHLST20}. 
While the synthesis problem becomes, in general, undecidable as soon as there 
are two universal quantifiers, there is a special class of universal 
specifications, called the linear $\forall^*$-fragment, which is still 
decidable. The linear $\forall^*$-fragment corresponds to the 
decidable \emph{distributed synthesis} problems~\cite{fs05}. The \emph{bounded 
synthesis} problem~\cite{DBLP:journals/acta/FinkbeinerHLST20,cfst19} considers only systems up to a given bound on the number of 
states. Bounded synthesis has been 
successfully applied to various benchmarks including the dining 
cryptographers~\cite{c85}.

The problem of \emph{model checking} hyperproperties for tree-shaped and acyclic graphs 
was studied in~\cite{bf18}. Earlier, a similar study of the impact of 
structural restrictions on the complexity of the model checking problem has 
also been carried out for LTL~\cite{kb11}.

Our motivating example draws from the substantial literature on specifying and
verifying \emph{non-repudiation protocols}~\cite{DBLP:conf/dbsec/EzhilchelvanS03,DBLP:conf/dsn/LiuNJ00,10.1007/3-540-44685-0_37,10.1007/978-3-642-29963-6_10}. In particular, \cite{10.1007/978-3-642-29963-6_10} discusses
the need for the consideration of incomplete information. We are not aware, however, of previous work on the automatic synthesis of trusted third parties
for such protocols.





\section{Conclusion and Future Work}
\label{sec:conclusion}

We have presented a rigorous classification of the complexity of the 
{\em controller synthesis} problem for {\em hyperproperties} expressed in 
HyperLTL. We considered general, acyclic, and tree-shaped plants. We showed 
that for trees, 
the complexity of the synthesis problem in the size of the plant does not go 
beyond \comp{NP}. While the problem is complete for \comp{L} for the 
alternation-free existential fragment, it is complete for \comp{NP} for the
alternation-free universal fragment. The problem is complete for \comp{P} for 
the fragment with only one quantifier alternation, where the leading 
quantifier is universal. For acyclic plants, the complexity is in 
\comp{PSPACE} (in the level of the polynomial hierarchy that corresponds to the 
number of quantifier alternations). The problem is \comp{NL-complete} for the 
alternation-free existential fragment. Similar to trees, the problem is 
\comp{NP-complete} for the alternation-free universal fragment. For general 
graphs, the problem is \comp{NONELEMENTARY} for an arbitrary number of 
quantifier alternations. For a bounded number $k$ of alternations, the problem 
is $(k{-}1)$-\comp{EXPSPACE-complete}.

It is interesting to compare controller synthesis to
program repair~\cite{bf18}. With the 
notable exception the universal fragment of HyperLTL for trees and acyclic 
graphs, the complexities of controller synthesis and program repair are largely 
aligned. This is mainly due to the fact that synthesizing a controller 
involves computing a subset of the controllable transitions such that the 
specification is satisfied. This is also the case for program repair, except that all transitions of the system are controllable.  

As for future work, we plan to develop efficient controller synthesis algorithms for different fragments of HyperLTL along the lines of QBF-based synthesis methods for hyperproperties~\cite{DBLP:journals/acta/FinkbeinerHLST20,cfst19}. It would furthermore be interesting 
to see if the differences we observed for HyperLTL carry over to other 
hyperlogics beyond HyperLTL (cf.~\cite{cfhh19,cfkmrs14,FinkbeinerMSZ-CCS17,ab18}).

\section*{Acknowledgments}

This work is sponsored in part by the United States NSF SaTC Award 1813388. It 
was also supported by the German Research Foundation (DFG) as part of the 
Collaborative Research Center “Methods and Tools for Understanding and 
Controlling Privacy” (CRC 1223) and the Collaborative Research Center 
“Foundations of Perspicuous Software Systems” (TRR 248, 389792660), and by the 
European Research Council (ERC) Grant OSARES (No. 683300).

\bibliographystyle{IEEEtran}
\bibliography{bibliography}

\end{document}